%% file: main.tex
\documentclass[pdflatex,iicol,sn-basic]{sn-jnl} 
\pdfoutput=1 

\usepackage{lipsum}
\usepackage{amsmath}
\usepackage{cool}
\usepackage{mathtools}
\usepackage{dirtytalk}
\usepackage{cuted}
\usepackage{relsize}
\usepackage[font=small]{caption}
\usepackage{subcaption}
\usepackage{tabularray}
\usepackage{xurl}

\algrenewcommand\algorithmicrequire{\textbf{Initialisation:}}

\usepackage{stfloats}

\usepackage{xcolor}
\usepackage[frozencache,cachedir=.]{minted} 
\usepackage{textgreek}
\DeclareUnicodeCharacter{2218}{$\scriptstyle\circ$}
\DeclareUnicodeCharacter{2248}{$\scriptstyle\approx$}
\DeclareUnicodeCharacter{2299}{$\scriptstyle\odot$}
\DeclareUnicodeCharacter{222B}{$\scriptstyle\int$}
\DeclareUnicodeCharacter{22C5}{$\scriptstyle\cdot$}
\DeclareUnicodeCharacter{2207}{$\scriptstyle\nabla$}
\DeclareUnicodeCharacter{1D39}{$^\texttt{M}$}
\DeclareUnicodeCharacter{1D34}{$^\texttt{H}$}

\newcommand{\jl}[1]{\small\Verb{#1}}

\definecolor{shadecolor}{gray}{.92}
\definecolor{incolor}{rgb}{0,0,.7}
\definecolor{outcolor}{rgb}{.65,0,0}
\definecolor{syntaxcolor}{rgb}{.65,0,0}
\definecolor{bg}{rgb}{0.93,0.93,0.93}
\definecolor{myblue}{RGB}{93,188,210}
\definecolor{mygreen}{RGB}{189,210,93}
\definecolor{myorange}{RGB}{210,173,93}
\definecolor{myred}{RGB}{210,93,130}
\definecolor{mydarkblue}{RGB}{93,130,210}
\definecolor{mydarkgreen}{RGB}{93,210,173}

\newminted{jl}{fontsize=\footnotesize,breaklines}

\jyear{2024}

\theoremstyle{thmstyleone}%
\newtheorem{theorem}{Theorem}
\newtheorem{proposition}[theorem]{Proposition}%

\theoremstyle{thmstyletwo}%

\theoremstyle{thmstylethree}%
\newtheorem{lemma}{Lemma}

\usepackage{ifthen}
\usepackage[normalem]{ulem}

\newboolean{AnnotateChanges}
\setboolean{AnnotateChanges}{false} 

\ifthenelse{\boolean{AnnotateChanges}}{%
\definecolor{annotateColor}{RGB}{0,0,255}
}{%
\definecolor{annotateColor}{RGB}{0,0,0} 
}

\newcommand{\added}[2][]{%
    \ifthenelse{\boolean{AnnotateChanges}}{%
        \textcolor{annotateColor}{#2}\textbf{\textcolor{red}{#1}}}{#2}}
        
\newcommand{\deleted}[2][]{%
    \ifthenelse{\boolean{AnnotateChanges}}{%
        \textcolor{annotateColor}{\ifmmode\text{\sout{\ensuremath{#2}}}\else\sout{#2}\fi}\textbf{\textcolor{red}{#1}}}{}}

\newcommand{\replaced}[3][]{%
    \ifthenelse{\boolean{AnnotateChanges}}{%
        \textcolor{annotateColor}{#2\sout{#3}}\textbf{\textcolor{red}{#1}}}{#2}}

\begin{document}

\title[\textcolor{white}{.}]{GridapTopOpt.jl: A scalable Julia toolbox for level set-based topology optimisation}

\author*[1]{\fnm{Zachary J.} \sur{Wegert}}\email{zach.wegert@hdr.qut.edu.au}
\author[2]{\fnm{Jordi} \sur{Manyer}}\email{jordi.manyer@monash.edu}
\author[3]{\fnm{Connor} \sur{Mallon}}\email{connor.mallon@monash.edu}
\author[2,4]{\fnm{Santiago} \sur{Badia}}\email{santiago.badia@monash.edu}
\author*[1]{\fnm{Vivien J.} \sur{Challis}}\email{vivien.challis@qut.edu.au}

\affil[1]{\orgdiv{School of Mathematical Sciences}, \orgname{Queensland University of Technology}, \orgaddress{\street{2 George St}, \city{Brisbane}, \postcode{4000}, \state{QLD}, \country{Australia}}}
\affil[2]{\orgdiv{School of Mathematics}, \orgname{Monash University}, \orgaddress{\street{Wellington Rd}, \city{Clayton}, \postcode{3800}, \state{Victoria}, \country{Australia}}}
\affil[3]{\orgdiv{Department of Chemical and Biological Engineering}, \orgname{Monash University}, \orgaddress{\street{Wellington Rd}, \city{Clayton}, \postcode{3800}, \state{Victoria}, \country{Australia}}}
\affil[4]{\orgdiv{Centre Internacional de M\'etodes Num\'erics a l’Enginyeria}, \orgname{Campus Nord}, \orgaddress{UPC, \postcode{08034}, \city{Barcelona}, \country{Spain}}}

\abstract{In this paper we present GridapTopOpt, an extendable framework for level set-based topology optimisation that can be readily distributed across a personal computer or high-performance computing cluster. The package is written in Julia and uses the Gridap package ecosystem for parallel finite element assembly from arbitrary weak formulations of partial differential equation (PDEs) along with the scalable solvers from the Portable and Extendable Toolkit for Scientific Computing (PETSc). The resulting user interface is intuitive and easy-to-use, allowing for the implementation of a wide range of topology optimisation problems with a syntax that is near one-to-one with the mathematical notation. Furthermore, we implement automatic differentiation to help mitigate the bottleneck associated with the analytic derivation of sensitivities for complex problems. GridapTopOpt is capable of solving a range of benchmark and research topology optimisation problems with large numbers of degrees of freedom. This educational article demonstrates the usability and versatility of the package by describing the formulation and step-by-step implementation of several distinct topology optimisation problems. The driver scripts for these problems are provided and the package source code is available at \url{https://github.com/zjwegert/GridapTopOpt.jl}.}


\maketitle
\textcolor{white}{a}
\newpage
\section{Introduction}\label{sec: intro}
Owing to a rich theoretical foundation and wide array of ever-evolving industrial applications, topology optimisation has become an increasingly popular computational approach in both industry and academia \citep{TopOptMonograph, DeatonGrandhi2013, TopOptReviewSigmund}. Introduced in the seminal paper by \cite{SeminalTop88}, the field now boasts several techniques for solving real-world problems. These include density-based methods \citep{Bendsoe89, Rozvanyetal1992} in which the design variables are material densities of elements or nodes in a mesh, and level set-based methods \citep{10.1016/S0045-7825(02)00559-5_2003,10.1016/j.jcp.2003.09.032_2004} in which the boundary of the domain is implicitly tracked via a level set function and updated according to an evolution equation. In conventional level set methods \cite[e.g.,][]{10.1016/S0045-7825(02)00559-5_2003,10.1016/j.jcp.2003.09.032_2004}, the Hamilton-Jacobi evolution equation is used to evolve the design interface under a normal velocity field that is inferred from shape derivatives. 

In topology optimisation, educational articles are an important way of disseminating techniques and code to students and newcomers to the field. Beginning with the \say{99-line} Matlab code by \cite{10.1007/s001580050176_2001} for density-based topology optimisation, there are now several educational contributions for density methods \citep[e.g.,][]{10.1007/s00158-010-0594-7_2011,10.1007/s00158-014-1157-0_2015,10.1007/s00158-021-03050-7_2021,10.1007/s00158-021-02917-z_2021,10.1007/s00158-022-03339-1_2022,10.1007/s00158-022-03420-9_2023}, level set methods \citep[e.g.,][]{10.1007/s00158-009-0430-0_2010,10.1007/s00158-014-1190-z_2015,10.1007/s00158-018-1950-2_2018,10.1007/s00158-018-1904-8_2018}, and other methods \citep[e.g.,][]{10.1007/s00158-020-02722-0_2021,10.1007/s00158-020-02816-9_2021,10.1007/s00158-023-03529-5_2023}. A detailed review of the educational literature is given by \cite{10.1007/s00158-021-03050-7_2021}, which makes clear that the vast majority of educational articles are implemented in Matlab or the open-source equivalent GNU Octave. In contrast, we utilise the open-source programming language Julia \citep{doi:10.1137/141000671} that provides several benefits in the context of topology optimisation. Firstly, Julia is implemented using a just-in-time (JIT) compiler that is able to convert generic high-level code to machine code that can be executed efficiently. As a result, users can implement type-generic algorithms and achieve similar performance to fully-compiled languages like C or Fortran while maintaining the flexibility of interpreted languages like MATLAB or Python. In addition, multiple dispatch allows the compiler to efficiently choose a function implementation based on argument type. This is particularly useful for extending a library without needing to resort to modifying the core implementation of the library. As we will later see, multiple dispatch also helps reduce code duplication particularly when writing code to be executed in parallel.

To make large-scale or multi-physics three-dimensional topology optimisation problems viable, it is becoming common practice to utilise parallel processing via distributed CPU (central processing unit) computing \citep[e.g.,][]{10.1007/s00158-014-1157-0_2015} or GPU (graphics processing unit) computing \cite[e.g.,][]{10.1137/070699822_2009,10.1007/s00158-013-0980-z_2014}. A common method for achieving a distributed CPU implementation is via a Message Passing Interface (MPI) \citep{mpi41}, whereby instances of a program are run across several CPUs and data is communicated between CPUs via MPI \citep{10.1007/s00158-014-1157-0_2015}. In this type of implementation, the mesh is typically partitioned into local parts and computation is conducted on each partition synchronously until a point where communication is required to pass data between partitions. For density-based methods there are a few educational articles that provide code and describe this computational framework \citep{10.1007/s00158-014-1157-0_2015,10.1007/s00158-021-02917-z_2021}. However, for level set-based methods there is a lack of educational articles pertaining to parallel computing methods. Note that there are non-educational articles that focus on this topic \cite[e.g.,][]{10.1016/j.cma.2018.08.028_2019,10.1016/j.compstruc.2019.05.010_2019,10.1016/j.jcp.2020.109574_2020,10.1016/j.finel.2021.103561_2021,10.1002/nme.6923_2022,10.1016/j.cma.2022.115112_2022,10.1007/s00158-021-03086-9_2022,Wu_Zhu_Gao_Gao_Liu_2024}, however none of these provide portable or easy-to-use code. In addition, distributed computing methods implemented in C\texttt{++} and the Portable, Extensible Toolkit for Scientific Computation (PETSc) are generally harder to extend owing to a low level application programming interface (API).

Motivated by the discussion above, the goal of this educational article and the Julia package \textit{GridapTopOpt} is to provide an extendable framework for level set-based topology optimisation with a high-level API that can be run in serial (single CPU) or distributed in parallel (multiple CPUs). To achieve this we utilise the Julia package Gridap \citep{Badia2020,Verdugo2022} and related satellite packages for the finite element (FE) computations. Gridap allows for the implementations of a weak formulation for a partial differential equation (PDE) problem with a syntax that is near one-to-one with the mathematical notation. Moreover, for the parallel implementation we leverage the distributed linear algebra package PartitionedArrays \citep{PartitionedArrays_Verdugo2021} along with the distributed finite element package GridapDistributed \citep{Badia2022}. Using these tools we implement parallel finite elements and parallel finite difference stencils for solving the Hamilton-Jacobi evolution equation while entirely avoiding the need for explicit MPI operations. In addition, we use the adjoint method and the automatic differentiation implemented in Gridap to provide automatic differentiation of PDE-constrained functionals. As a result of the aformentioned design choices, GridapTopOpt is able to handle a wide range of topology optimisation problems (e.g., two/three dimensional, linear/nonlinear, etc.), can be run in either serial or parallel with minor code modifications, and possesses a high level API that facilitates fast user implementation of new problems with automatic differentiation. We expect GridapTopOpt to be a powerful package for solving new topology optimisation problems.

The remainder of the paper is as follows. In Section \ref{sec: mathback} we discuss the mathematical background for level set-based topology optimisation and automatic differentiation. In Section \ref{sec: impl} we discuss the implementation of GridapTopOpt. In Section \ref{sec: examples} we provide useful educational examples and demonstrate the versatility of our implementation by formulating and solving several topology optimisation problems using the framework. In Section \ref{sec: benchmarks} we provide strong and weak scaling benchmarks of the parallel implementation. Finally, in Section \ref{sec: concl} we present our concluding remarks.

\section{Mathematical background}\label{sec: mathback}
In this section we present a general description of a topology optimisation problem and subsequently outline the mathematical concepts used in the implementation of GridapTopOpt.

\subsection{General problem description}\label{sec: gen}
Suppose we consider a bounding domain $D\subset\mathbb{R}^d$ with $d=2$ or $d=3$. We consider finding a subdomain $\Omega\subset D$ that solves the PDE-constrained optimisation problem:
\begin{equation}\label{eqn: generic opt prob}
\begin{aligned}
    \underset{\Omega\subset D}{\min}~
      &J(\Omega,u(\Omega))\\
    \text{s.t.}~~ & C_i(\Omega,u(\Omega))=0,~i=1,\dots,N,\\
    &a({u},{v})=l({v}),~\forall {v}\in V,
\end{aligned}
\end{equation}
where the objective $J$ and constraints $C_i$ depend on the solution $u(\Omega)$ of the PDE state equations represented in weak form in the final line. In the following we omit $u(\Omega)$ from $J$ and $C_i$ for notational brevity. Note that in this formulation $u$ can be scalar-valued, vector-valued, or a collection of several states in the case of multi-physics problems. In specific concrete problem descriptions we will denote vector-valued states by a bold symbol.

\subsection{Shape differentiation}\label{sec: shape diff}
The first-order change of a functional $J(\Omega)$ under a change in $\Omega$ can be understood by considering smooth variations of the domain $\Omega$ of the form $\Omega_{\boldsymbol{\theta}} =(\boldsymbol{I}+\boldsymbol{\theta})(\Omega)$, where $\boldsymbol{\theta} \in W^{1,\infty}(\mathbb{R}^d,\mathbb{R}^d)$ is a vector field. The shape derivative of $J(\Omega)$ at $\Omega$ is defined as the Fr\'echet derivative in $W^{1, \infty}(\mathbb{R}^d, \mathbb{R}^d)$ at $\boldsymbol{\theta}$ of the application $\boldsymbol{\theta} \rightarrow J(\Omega_{\boldsymbol{\theta}})$, i.e.,
\begin{equation}\label{eqn: shape deriv defin}
J(\Omega_{\boldsymbol{\theta}})=J(\Omega)+J^{\prime}(\Omega)(\boldsymbol{\theta})+\mathrm{o}(\boldsymbol{\theta})  
\end{equation}
with $\lim _{\boldsymbol{\theta} \rightarrow 0} \frac{\lvert\mathrm{o}(\boldsymbol{\theta})\rvert}{\|\boldsymbol{\theta}\|}=0,$ where the shape derivative $J^{\prime}(\Omega)$ is a continuous linear form on $W^{1, \infty}(\mathbb{R}^d, \mathbb{R}^d)$ \citep{10.1016/j.jcp.2003.09.032_2004,10.1016/bs.hna.2020.10.004_978-0-444-64305-6_2021}.

The shape derivatives of volume and surface functionals whose integrands do not depend on $\Omega$ are readily available in the literature \citep[e.g.,][]{ChallisThesis,10.1016/bs.hna.2020.10.004_978-0-444-64305-6_2021}. When the functional of interest $J$ depends on the solution of a PDE at $\Omega$, C\'ea's formal method \citep{10.1051/m2an/1986200303711_1986} can be applied to find the shape derivative of $J$. A detailed discussion of this method is given by \cite{10.1016/bs.hna.2020.10.004_978-0-444-64305-6_2021}.

\subsection{Descent direction}\label{sec: descent dir}
As discussed by \cite{10.1016/bs.hna.2020.10.004_978-0-444-64305-6_2021}, it is common that the shape derivative of $J$ takes the form
\begin{equation}\label{eqn: speical shape der}
    J^\prime(\Omega)(\boldsymbol{\theta})=\int_{\partial\Omega}q~\boldsymbol{\theta}\cdot\boldsymbol{n}~\mathrm{ds}
\end{equation}
where $q:\partial\Omega\rightarrow\mathbb{R}$ is sometimes called the \textit{shape sensitivity}. A possible descent direction that reduces $J$ in Eq. \ref{eqn: shape deriv defin} is then $\boldsymbol{\theta}=-\tau q\boldsymbol{n}$ for sufficiently small $\tau>0$. 

It is common to utilise the Hilbertian extension-regularisation framework (see \cite{10.1016/bs.hna.2020.10.004_978-0-444-64305-6_2021} and references therein) to solve an identification problem over a Hilbert space $H$ on $D$ with inner product $\langle\cdot,\cdot\rangle_H$: 
\begin{equation}\label{eqn: hilb extension wf}%
   \begin{aligned}
        &\textit{Find }g_\Omega\in H\textit{ such that}\\
        &\langle g_\Omega,w\rangle_H=J^{\prime}(\Omega)(-w\boldsymbol{n})~\forall w\in H.
    \end{aligned}
\end{equation}%
The solution $g_\Omega$ to this problem naturally extends the shape sensitivity from $\partial\Omega$ onto the bounding domain $D$, and ensures a descent direction for $J(\Omega)$ with $H$-regularity. Importantly, using the definition of a shape derivative, it can be shown that
\begin{align*}
    J(\Omega_{\boldsymbol{\theta}})<J(\Omega)
\end{align*}
where $\tau>0$ is sufficiently small.

In this work we choose $H=H^1(D)$ with the inner product 
\begin{equation}\label{eqn: H1 inner product}
    \langle u,v\rangle_{H}=\int_D\left(\alpha^2\boldsymbol{\nabla} u\cdot\boldsymbol{\nabla} v+uv\right)\mathrm{d}\boldsymbol{x},
\end{equation}
where $\alpha$ is the so-called regularisation length scale \cite[e.g.,][]{10.1007/s00158-016-1453-y_2016,10.1007/s40324-018-00185-4_2019,10.1007/s10957-021-01928-6_2021}.

\subsection{Constrained optimisation}
In this paper we implement the classical augmented Lagrangian method to solve the constrained optimisation problem posed in Eq. \ref{eqn: generic opt prob}. As per the discussion by \cite{10.1016/bs.hna.2020.10.004_978-0-444-64305-6_2021} and more generally by \cite{978-0-387-30303-1_2006}, we consider the unconstrained optimisation problem
\begin{equation}\label{eqn: ALM opt prob}
\begin{aligned}
    \underset{\Omega\subset D}{\min}~
      &\mathcal{L}(\Omega)\coloneqq J(\Omega)-\sum_{i=1}^{N}\left\{\lambda_iC_i(\Omega)+\frac{1}{2}\Lambda_iC_i(\Omega)^2\right\}
\end{aligned}
\end{equation}
where $\mathcal{L}$ is called the augmented Lagrangian functional and has shape derivative
\begin{equation}
    \mathcal{L}^\prime(\Omega)(\boldsymbol{\theta}) = J^\prime(\Omega)(\boldsymbol{\theta})-\sum_{i=1}^{N}\left[\lambda_i-\Lambda_iC_i(\Omega)\right]C_i^\prime(\Omega)(\boldsymbol{\theta}).
\end{equation}
In the above $\lambda_i$ and $\Lambda_i$ are called the Lagrange multipliers and penalty parameters, respectively. A new estimate of the Lagrange multipliers $\lambda_i$ at iteration $k$ is given by
\begin{equation}
    \lambda_i^{k+1}=\lambda_i^k-\Lambda_iC_i(\Omega^k)
\end{equation}
where $\Omega^k$ is the domain at iteration $k$. In addition, the penalty parameters $\Lambda_i$ are updated periodically according to
\begin{equation}
    \Lambda_i^{k+1}= \min(\xi\Lambda_i^{k},\Lambda_\textrm{max})
\end{equation}
for $\xi>1$. We can then determine a new descent direction at iteration $k$ using the Hilbertian extension-regularisation framework. That is, we solve: \textit{Find $l_\Omega^{k+1}\in H$ such that}
\begin{equation}\label{eqn: ALM ext shape sens}
    \langle l_\Omega^{k+1},w\rangle_H=\mathcal{L}^{\prime}(\Omega^k)(-w\boldsymbol{n})~\forall w\in H
\end{equation}
to determine a new descent direction $l_\Omega^{k+1}$ that improves the augmented Lagrangian $\mathcal{L}$.

We note that there are several other methods that consider constrained topology optimisation such as Lagrangian methods \citep{10.1016/j.jcp.2003.09.032_2004,10.1007/s00158-018-1950-2_2018}, sequential linear programming methods \citep{10.1179/1743284715Y.0000000022_2015,10.1007/s00158-014-1174-z_2015,10.1016/bs.hna.2020.10.004_978-0-444-64305-6_2021}, and null space/projection methods \citep{10.3934/dcdsb.2019249,10.1051/cocv/2020015_2020,10.1007/s00158-023-03663-0_2023}. Our package GridapTopOpt also implements a Hilbertian projection method as described by \cite{10.1007/s00158-023-03663-0_2023}.

\subsection{The level set method}\label{sec: lsm}
We consider tracking and updating a subdomain $\Omega$ over a fixed computational domain $D$ using a level set function $\phi:D\rightarrow\mathbb{R}$ \citep{978-0-521-57202-6_1996,978-0-387-22746-7_2006}. This approach corresponds to the so-called \textit{Eulerian} shape capturing approach as discussed by \cite{10.1016/bs.hna.2020.10.004_978-0-444-64305-6_2021}. Other methods such as the \textit{Lagrangian} shape capturing approach use re-meshing algorithms to exactly discretise an interface \citep{10.1016/j.cma.2014.08.028_2014}. These are particularly useful for fluid-structure interaction problems.

\subsubsection{Domain tracking and integration}
The level set function $\phi$ is typically defined as 
\begin{equation}\label{eqn: lsf def}
\begin{cases}\phi(\boldsymbol{x})<0&\text{if } \boldsymbol{x} \in \Omega, \\ \phi(\boldsymbol{x})=0 &\text{if } \boldsymbol{x} \in \partial \Omega, \\\phi(\boldsymbol{x})>0 &\text{if } \boldsymbol{x} \in D \backslash \bar{\Omega}.\end{cases}
\end{equation}
As discussed by \cite{978-0-387-22746-7_2006}, this definition readily admits a method for integrating a function over $\Omega$ and $\partial\Omega$ in the fixed computational regime. Namely, given a function $f:\mathbb{R}^d\rightarrow\mathbb{R}$ we have:
\begin{equation}\label{eqn: lsf integration1}
    \int_{\Omega}f(\boldsymbol{x})\,\mathrm{d}\boldsymbol{x}=\int_{D}f(\boldsymbol{x})(1-H(\phi))\,\mathrm{d}\boldsymbol{x},
\end{equation}
and
\begin{equation}\label{eqn: lsf integration2}
    \int_{\partial\Omega}f(\boldsymbol{x})\,\mathrm{ds}=\int_{D}f(\boldsymbol{x})H^\prime(\phi)\lvert\boldsymbol{\nabla}\phi\rvert\,\mathrm{d}\boldsymbol{x},
\end{equation}
where $H$ is a Heaviside function defined as 
\begin{equation}
    H(\phi)=
    \begin{cases}
        0&\text{if }\phi\leq0,\\
        1&\text{if }\phi>0.
    \end{cases}
\end{equation}

It is common to avoid numerical and integrability issues in the above by approximating the Heaviside function $H$ with a smooth Heaviside function $H_\eta$. A possible choice for this function is given by \citep{978-0-387-22746-7_2006}
\begin{equation}\label{eqn: smoothed heavi}
    H_\eta(\phi)=\begin{cases}
    0&\text{if }\phi<-\eta,\\
    \frac{1}{2}+\frac{\phi}{2\eta}+\frac{1}{2\pi}\sin\left(\frac{\pi\phi}{\eta}\right)&\text{if }\lvert\phi\rvert\leq\eta,\\
    1&\text{if }\phi>\eta,
    \end{cases}
\end{equation}
where $\eta$ is half the width of the small transition region between 0 and 1. 

In practice, using Eq. \ref{eqn: lsf integration1} for the solution of state equations in a finite element method yields a singular stiffness matrix \citep{10.1016/j.jcp.2003.09.032_2004}. A popular method to deal with this situation is the ersatz material approximation \citep[e.g.,][]{10.1016/j.jcp.2003.09.032_2004} where the void phase is filled with a soft material so that the state equations can be solved without a body-fitted mesh. To this end, we also introduce a smooth characteristic function $I:\mathbb{R}\rightarrow[\epsilon,1]$ defined as
\begin{equation}\label{eqn: erstaz mat}
    I(\phi)=(1-H_\eta(\phi))+\epsilon H_\eta(\phi),
\end{equation}
where $\epsilon\ll1$. This will be used in Section \ref{sec: examples}. In future we hope to better track the domain boundary using unfitted finite element methods.

\subsubsection{Level set update}\label{sec: ls update}
To evolve the level set function under a normal velocity $v(\boldsymbol{x})$, we solve the Hamilton-Jacobi evolution equation \citep{978-0-521-57202-6_1996,978-0-387-22746-7_2006,10.1016/bs.hna.2020.10.004_978-0-444-64305-6_2021}:
\begin{align}\label{eqn: HJ}
    \begin{cases}
    \pderiv{\phi}{t}(t,\boldsymbol{x}) + v(\boldsymbol{x})\lvert\boldsymbol{\nabla}\phi(t,\boldsymbol{x})\rvert = 0,\\
    \phi(0,\boldsymbol{x})=\phi_0(\boldsymbol{x}),\\
    \boldsymbol{x}\in D,~t\in(0,T),
    \end{cases}
\end{align}
where $T>0$ is small and $\phi_0$ is the initial state of $\phi$ at $t=0$. In the context of topology optimisation using the augmented Lagrangian method, the normal velocity $v$ is given by the extended shape sensitivity $l_\Omega^{k}$ as defined by Eq. \ref{eqn: ALM ext shape sens} at an intermediate iteration $k$. Solving Eq. \ref{eqn: HJ} using this choice of normal velocity provides a new domain that reduces the Lagrangian to first order.

\subsubsection{Level set reinitialisation}\label{sec: ls reinit}
To ensure that the level set function is neither too steep nor too flat near the boundary of $\Omega$, it is useful to reinitialise the level set function as a signed distance function $d_\Omega$ \citep{978-0-387-22746-7_2006}. For the domain $\Omega$ the signed distance function is defined as \citep{10.1016/bs.hna.2020.10.004_978-0-444-64305-6_2021}:
\begin{equation}
d_{\Omega}(\boldsymbol{x})=\begin{cases}
-d(\boldsymbol{x}, \partial \Omega)&\text{if }\boldsymbol{x} \in \Omega, \\
0&\text{if } \boldsymbol{x} \in \partial \Omega, \\
d(\boldsymbol{x}, \partial \Omega)&\text{if } \boldsymbol{x} \in D \backslash \bar{\Omega},
\end{cases}
\end{equation}
where $d(\boldsymbol{x}, \partial \Omega):=\min _{\boldsymbol{p} \in \partial \Omega}\lvert\boldsymbol{x}-\boldsymbol{p}\rvert$ is the minimum Euclidean distance from $\boldsymbol{x}$ to the boundary $\partial \Omega$. 

In this work we use the reinitialisation equation to reinitialise a pre-existing level set function $\phi_0(\boldsymbol{x})$ as the signed distance function \citep{10.1006/jcph.1999.6345_1999,978-0-387-22746-7_2006}. This is given by
\begin{align}\label{eqn: reinit}
\begin{cases}
    \pderiv{\phi}{t}(t,\boldsymbol{x}) + S(\phi_0(\boldsymbol{x}))\left(\lvert\boldsymbol{\nabla}\phi(t,\boldsymbol{x})\rvert-1\right) = 0,\\
    \phi(0,\boldsymbol{x})=\phi_0(\boldsymbol{x}),\\
    \boldsymbol{x}\in D, t>0.
\end{cases}
\end{align}
where $S$ is the sign function. In this work we approximate $S$ using
\begin{equation}\label{eqn: approx sign}
    S(\phi)=\frac{\phi}{\sqrt{\phi^2+\epsilon_s^2\lvert\nabla\phi\rvert^2}},
\end{equation}
where $\epsilon_s$ is the minimum element side length \citep{10.1006/jcph.1999.6345_1999}. Note that $S(\phi)$ should be updated at every time-step when solving Eq. \ref{eqn: reinit} \citep{978-0-387-22746-7_2006}.

\subsection{Automatic differentiation}\label{sec: ad impl}
Owing to our use of a fixed computational domain $D$, we implement automatic differentiation with respect to the level set function $\phi$ using the adjoint method \citep[][Sec. 5.2.5]{10.1093/acprof:oso/9780199546909.001.0001_2009}. This is closely related to C\'ea's method \citep{10.1051/m2an/1986200303711_1986}.
In particular, we consider differentiation of a functional $F(u(\phi),\phi)$ with respect to $\phi$ where $u$ satisfies a PDE residual $\mathcal{R}(u(\phi),\phi)=0$. The derivative $\D{F}{\phi}$ is then given by
\begin{equation}
    \D{F}{\phi}=\pderiv{F}{\phi}+\pderiv{F}{u}\pderiv{u}{\phi}.
\end{equation}
where differentiation through the residual $\mathcal{R}$ yields $\pderiv{F}{u}\pderiv{u}{\phi}=-\lambda^\intercal\pderiv{\mathcal{R}}{\phi}$ for $\lambda$ satisfying the adjoint equation
\begin{equation}\label{eqn: adjoint eqn}
    \pderiv{\mathcal{R}}{u}^\intercal\lambda=\pderiv{F}{u}^\intercal.
\end{equation}
By leveraging the automatic differentiation capabilities implemented in  the Julia finite element package Gridap \citep{Verdugo2022} through ForwardDiff \citep{RevelsLubinPapamarkou2016}, all the additional quantities (e.g., $\pderiv{\mathcal{R}}{u}$) are readily computed. Note that the above can readily be expressed in terms of Fr\'echet derivatives in $\phi$, however we do not do this here to avoid introducing further notation.

This discussion corresponds to a `discretise-then-optimise' paradigm \citep{10.1016/bs.hna.2020.10.004_978-0-444-64305-6_2021}, as opposed to the mathematical discussion in Section \ref{sec: shape diff} where shape derivatives are first found analytically.
It has been shown previously by \cite{10.1016/S0045-7825(02)00559-5_2003} that Fr\'echet derivatives in $\phi$ provide suitable descent directions via velocity extension methods similar to the Hilbertian extension-regularisation method discussed in Section \ref{sec: descent dir}.
In fact one can show the following result:
\begin{proposition}
    Suppose that we have $\Omega\subset D$ and $f\in W^{1,1}(\mathbb{R}^d)$. In addition, suppose that we have a level set function $\phi$ that satisfies $\lvert\boldsymbol{\nabla}\phi\rvert=1$ almost everywhere on $D$. Define
    \begin{equation}\label{prop: eqn1}
        J(\Omega)=\int_\Omega f~\mathrm{d}\boldsymbol{x}
    \end{equation}
    and the corresponding relaxation over $D$ as
    \begin{equation}\label{prop: eqn2}
        \hat{J}(\phi)=\int_D (1-H(\phi))f~\mathrm{d}\boldsymbol{x}.
    \end{equation}
    Then $J^\prime(\Omega)(-v\boldsymbol{n})=\hat{J}^\prime(\phi)(v)$ when the shape derivative of $J$ is relaxed over $D$.
    \end{proposition}
    \begin{proof}
    Applying Lemma 4 \citep{10.1016/j.jcp.2003.09.032_2004} to Eq. \ref{prop: eqn1}, relaxing the shape derivative over $D$ using Eq. \ref{eqn: lsf integration2}, and taking $\boldsymbol{\theta}=-v\boldsymbol{n}$ gives 
    \begin{equation}\label{eqn:pf1}
        J^\prime(\Omega)(-v\boldsymbol{n})=-\int_{D}vfH^\prime(\phi)\lvert\boldsymbol{\nabla}\phi\rvert~\mathrm{d}\boldsymbol{x}.
    \end{equation}
    On the other hand, a Fr\'echet derivative of $\hat{J}$ in the direction $v$ gives
    \begin{equation}\label{eqn:pf2}
        \hat{J}^\prime(\phi)(v)=-\int_D v f H^\prime(\phi)~\mathrm{d}\boldsymbol{x}.
    \end{equation}
    Equality of Eq. \ref{eqn:pf1} and \ref{eqn:pf2} follows when $\lvert\boldsymbol{\nabla}\phi\rvert=1$ on $D$ except possibly on a set $\Sigma$ of measure zero.
\end{proof}
\noindent Therefore, it is possible to recover the shape derivative $J^\prime$ when relaxed over $D$ using automatic differentiation to compute $\hat{J}^\prime$ as described above. We have also verified this result numerically. The analogous result does not hold for integrals over $\partial\Omega$ (e.g., Lemma 5, \cite{10.1016/j.jcp.2003.09.032_2004}) and a descent direction is not immediately clear in this case. This is remedied by the Hilbertian extension-regularisation method as it provides a descent direction that reduces the objective to first order.

It should be noted that true automatic shape differentiation is possible over a body-fitted mesh \citep[e.g.,][]{Ham_Mitchell_Paganini_Wechsung_2019,Gangl_Sturm_Neunteufel_Schöberl_2021}. In future we plan to investigate the alternate use of unfitted finite element methods for the automatic calculation of the shape derivatives as in \cite{mallon2024neural}, where, similar to using a body-fitted mesh, the domain is segmented by an infinitely sharp interface.

\section{Implementation}\label{sec: impl}

In the following we provide a high-level discussion on the implementation of GridapTopOpt. We begin by discussing distributed computing then discuss package design and installation.

\subsection{Distributed computing}
To efficiently tackle the high mesh resolutions required by real-world applications, our code makes use of distributed-memory parallelism. Within this computational paradigm, each processor holds a portion of the computational domain and performs local operations on it. Global behavior is achieved by sharing local information between processors through so-called messages. The Message Passing Interface (MPI) \citep{mpi41} provides a well-established standard for these message-passing operations, and is supported by Julia through the package MPI.jl \citep{Byrne2021}.

Due to the reduced compact support of most functional basis used in FEs, this parallel programming model is quite well suited for the FEM framework. Indeed, most computations (i.e creation of the FE spaces, evaluation of FE functions, and integration of weak forms) only require access of information from the local portion of the mesh and a single-depth layer of cells belonging to other processors. This extra layer, called the ghost layer, provides the global consistency.

In this context, Gridap provides abstract interfaces for the FE solution of PDEs, as well as serial implementations of these abstractions. GridapDistributed \citep{Badia2022} implements these abstractions for distributed-memory environments, by leveraging the serial code in Gridap to handle the local portion on each parallel task. This is all possible thanks to the multiple-dispatch \citep{doi:10.1137/141000671} features available in Julia, which allow the compiler to efficiently choose the function implementation based on argument type (see examples below).

PartitionedArrays \citep{PartitionedArrays_Verdugo2021} is used to handle global data distribution layout and communication among tasks, without having to resort to the low-level MPI interface provided by Julia. PartitionedArrays also provides a parallel implementation of partitioned global linear systems (i.e., linear algebra vectors and sparse matrices).

The most important abstraction provided by PartitionedArrays is the concept of distributed arrays, where each index of the array belongs to a different processor. Within this abstraction, any distributed object can be though of as a wrapper around a distributed array of serial objects and a glue that allows these serial objects to stay consistent, like shown in the following snippet:

\begin{jlcode}
struct Object
  ## fields
end

struct DistributedObject
  objects :: MPIArray{Object}
  glue
end
\end{jlcode}

For instance, a \jl{DistributedDiscreteModel} contains a distributed array of \jl{DiscreteModel}s and a distributed global numbering for the cells in each local portion, represented by a \jl{PRange}. Similarly, a \jl{PVector} (i.e distributed vector) contains a distributed array of \jl{Vector}s and a distributed global numbering for the vector indices.

An important programming pattern is then provided by the \jl{map} function, which allows running serial code on each local portion in parallel, potentially followed by communications to keep a global consistency. Thanks to multiple-dispatch, the same method name can be used for both serial and parallel implementations, allowing for the following snippet:

\begin{jlcode}
function op(obj::Object)
  # Local operations
  return result
end

function op(obj::DistributedObject)
  result = map(op,local_views(obj))
  # Additional parallel operations, 
  #  if required
end
\end{jlcode}

In the above example, the function \jl{op} has two implementations. The first one takes a serial object \jl{obj} and does the work required within each processor without any communication involved. The second expects a distributed object \jl{obj}, containing a distributed array of \jl{Object} that can be retrieved through the method \jl{local_views}. The serial code is then run on the processor by mapping the operation onto each local portion, which will return a distributed array containing the serial results in each processor. If necessary, global sharing of information is performed afterwards.

\subsection{Package design}\label{sec: pkg design}
GridapTopOpt is designed to be a readily extendable object-oriented toolbox that can be called in serial or parallel by a user via a \textit{driver} script (see Appendix \ref{sec: code appendix}). Figure \ref{fig:core-package-structure} shows an illustration of this workflow along with the breakdown of core source files. 
\begin{figure}[t]
    \centering
    \includegraphics[width=0.6\columnwidth]{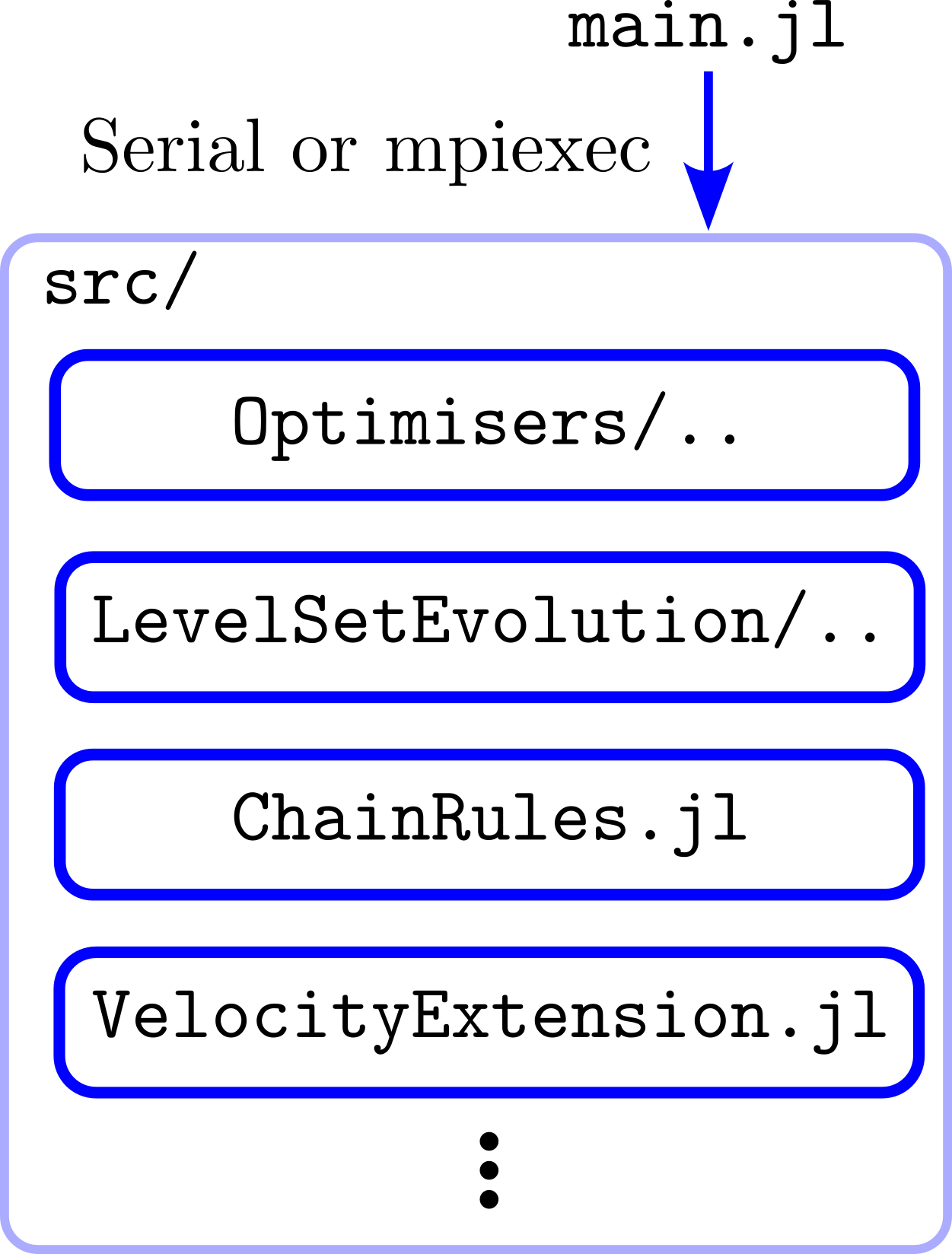}
    \caption{The package design of GridapTopOpt. The user-provided driver script \texttt{main.jl} implements the optimisation problem using the API provided in \texttt{src/}. The driver can then be run in serial or parallel.}
    \label{fig:core-package-structure}
\end{figure}
We segment the source code into distinct blocks so that new functionality can be added by modifying the source code in a file or alternatively outside the package by extending abstract interfaces. The source file contents are organised as follows:
\begin{itemize}
    \item[--] \texttt{Optimisers/..:} Contains an implementation of the augmented Lagrangian method \citep{978-0-387-30303-1_2006} as described in Section \ref{sec: descent dir} as well as an implementation of the Hilbertian projection method \citep{10.1007/s00158-023-03663-0_2023}. In addition, we include abstract types and interfaces that can be implemented to create a new optimiser.
    \item[--] \texttt{LevelSetEvolution/..:} Solves the Hamilton-Jacobi evolution equation and reinitialisation equation in parallel on $n$th order finite elements using a first-order upwind finite difference scheme \citep{10.1006/jcph.1999.6345_1999,10.1016/j.jcp.2003.09.032_2004,978-0-387-22746-7_2006}. We also provide abstract types and interfaces that can be implemented to create a new level set evolution method (e.g., Reaction-Diffusion-based methods).
    \item[--] \texttt{ChainRules.jl:} Implements the solution to the state equations for different types of residuals (e.g., affine or nonlinear) and enables automatic differentiation of optimisation functionals.
    \item[--] \texttt{VelocityExtension.jl:} Implements the Hilbertian extension-regularisation method \citep{10.1016/bs.hna.2020.10.004_978-0-444-64305-6_2021} for a user-defined Hilbert space and inner product.
\end{itemize}
Note that we omit source files from this list that provide additional utility (e.g., \texttt{Benchmarks.jl}). Additional information regarding these files can be found in the package documentation. We further note that all source code is written in Julia and does not resort to using other programming languages, thus avoiding the two-language barrier often associated with high performance computing \citep{doi:10.1137/141000671}. In addition, automatic differentiation in multi-language implementations becomes difficult or even impossible as the propagation of dual numbers through codes is interrupted without a proper interface.

\subsection{Package installation}
GridapTopOpt can be installed in an existing Julia environment using the package manager. This can be accessed in the Julia REPL (read-eval–print loop) by pressing \jl{]}. We then add the required core packages via:
\begin{jlcode}
pkg> add GridapTopOpt, Gridap
\end{jlcode}
Once installed, serial driver scripts can be run immediately, whereas parallel problems also require an MPI installation and some additional packages. For basic users, MPI.jl provides such an implementation and a Julia wrapper for \jl{mpiexec} - the MPI executor. This is installed via:
\begin{jlcode}
pkg> add MPI, GridapDistributed, GridapPETSc, GridapSolvers, PartitionedArrays, SparseMatricesCSR
julia> using MPI
julia> MPI.install_mpiexecjl()
\end{jlcode}
Once the \jl{mpiexecjl} wrapper has been added to the system \jl{PATH}, MPI scripts can be executed in a terminal via
\begin{jlcode}
mpiexecjl -n P julia main.jl
\end{jlcode}
where \jl{main} is a \textit{driver} script, and \jl{P} denotes the number of processors.

\section{Numerical examples}\label{sec: examples}

In this section we demonstrate the capabilities of GridapTopOpt by describing the mathematical formulation and driver scripts for several example topology optimisation problems. The driver scripts for all problems are given in the GitHub repository where our source code is readily available. See the \hyperref[rep results]{Replication of Results} section at the end of the manuscript for details.

\subsection{Minimum thermal compliance}
In our first example we consider minimising the thermal compliance (or dissipated energy) of heat transfer through a domain $\Omega$. A similar problem was considered by \cite{10.1016/j.cma.2006.08.005_2007}. 
\subsubsection{Formulation}
Suppose that we consider a two-dimensional problem on $D=[0,1]^2$. We can write down the state equations describing this problem using the homogeneous steady-state heat equation as
\begin{equation}\label{eqn: heat eqn}
\left\{\begin{aligned}
-\boldsymbol{\nabla}(\kappa\boldsymbol{\nabla} u) &= 0~\text{in }\Omega,\\
\kappa\boldsymbol{\nabla} u\cdot\boldsymbol{n} &= g~\text{on }\Gamma_N,\\
\kappa\boldsymbol{\nabla} u\cdot\boldsymbol{n} &= 0~\text{on }\Gamma,\\
u &= 0~\text{on }\Gamma_D.
\end{aligned}\right.
\end{equation}
where $\kappa$ is the thermal diffusivity within $\Omega$ and $\boldsymbol{n}$ is the unit normal on the boundary $\partial\Omega$. Physically we can imagine this as describing the transfer of heat through the domain $\Omega$ from the source to sinks under the following boundary conditions:
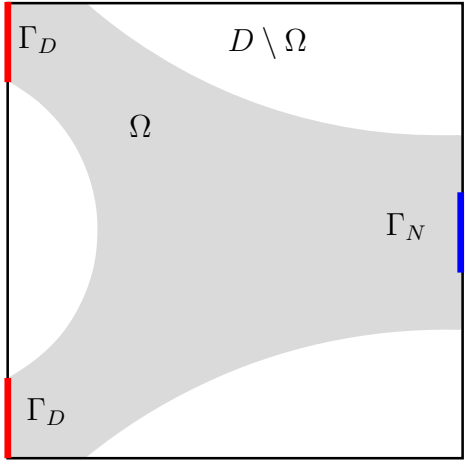
\begin{figure}[t]
    \centering
    \def\svgwidth{0.8\columnwidth}
    \large
    \input{Figure2}
    \caption{A schematic for the two-dimensional minimum thermal compliance problem.}
    \label{fig:thermal2d-setup}
\end{figure}
\begin{itemize}
    \item[--] \textit{Heat source:} normal heat flow $g$ across fixed $\Gamma_{N}$.
    \item[--] \textit{Isothermal:} zero heat on fixed $\Gamma_D$.
    \item[--] \textit{Adiabatic:} zero normal heat flow across $\Gamma$ where $\Gamma\coloneqq\partial\Omega\setminus(\Gamma_N\cup\Gamma_D)$,
\end{itemize} 
The weak formulation of the above strong formulation can be found by multiplying by a test function $v$ and applying integration by parts. The weak formulation is given by: \textit{Find} $u\in H^1_{\Gamma_D}(\Omega)$ \textit{such that}
\begin{equation}\label{eqn: thermal weak form}
\int_{\Omega}\kappa\boldsymbol{\nabla}u\cdot\boldsymbol{\nabla}v~\mathrm{d}\boldsymbol{x} = \int_{\Gamma_N}gv~\mathrm{d}s,~\forall v\in H^1_{\Gamma_D}(\Omega)
\end{equation}
\textit{where} $H^1_{\Gamma_D}(\Omega)=\{v\in H^1(\Omega):~v=0\textit{ on }\Gamma_D\}$. 

The minimum thermal compliance optimisation problem is then given by
\begin{equation}
    \begin{aligned}
    \min_{\Omega\in{D}}&~J(\Omega)\coloneqq\int_{\Omega}\kappa\boldsymbol{\nabla}u\cdot\boldsymbol{\nabla}u~\mathrm{d}\boldsymbol{x}\\
    \text{s.t. }&~C(\Omega)=0,\\
    &~a(u,v)=l(v),~\forall v\in H^1_{\Gamma_D}(\Omega),
    \end{aligned}
\end{equation}
where $C(\Omega)\coloneqq\operatorname{Vol}(\Omega)-V_f\operatorname{Vol}(D)$ with $\operatorname{Vol}(\Omega)=\int_\Omega1~\mathrm{d}\boldsymbol{x}$ and similarly for $\operatorname{Vol}(D)$, and $V_f$ is the specified required volume fraction. We introduce a volume-fraction constraint because the solution to the unconstrained problem is a solid plate. It can be shown that the shape derivatives of $J(\Omega)$ and $\operatorname{Vol}(\Omega)$ are given respectively by
\begin{equation}\label{eqn: shape derive thermal comp}
J'(\Omega)(\boldsymbol{\theta}) = -\int_{\Gamma}\kappa\boldsymbol{\nabla}u\cdot\boldsymbol{\nabla}u~\boldsymbol{\theta}\cdot\boldsymbol{n}~\mathrm{d}s
\end{equation}
and
\begin{equation}\label{eqn: shape deriv vol}
\operatorname{Vol}'(\Omega)(\boldsymbol{\theta}) = \int_{\Gamma}\boldsymbol{\theta}\cdot\boldsymbol{n}~\mathrm{d}s.
\end{equation}
We show the first of these in Appendix \ref{sec: thermal shape deriv appendix} using C\'ea's method, while the latter follows from Lemma 4 \citep{10.1016/j.jcp.2003.09.032_2004}. Note that we assume $\Gamma_N$ and $\Gamma_D$ are fixed.

The problem can now be formulated in GridapTopOpt. The following sections outline the driver script provided in Appendix \ref{sec: code appendix}.

It should be noted that this problem can alternatively be posed as a multi-phase problem with a volumetric heat source over all $D$ and with $\Omega$ and $D\setminus \Omega$ containing strongly and weakly conductive material, respectively. Such multi-phase problems require special treatment in level set-based topology optimisation and we refer to \cite{10.1051/cocv/2013076_2014} for further details.

\subsubsection{Initialisation and parameters (lines 1--27)}
In line 1 we load the GridapTopOpt and Gridap packages. These provide all required functionality for the driver script as provided in Appendix \ref{sec: code appendix}. Additional packages will be required when we consider extensions of the driver.

The parameters for the finite element discretisation are given in lines 5--14. We start by specifying the order of the finite element space using \jl{order}. The  size of the bounding domain and number of elements are then specified by the tuples \jl{dom} and \jl{el_size}, respectively. Lines 9--14 then specify parameters (\jl{prop_Γ_N} \& \jl{prop_Γ_D}) and Boolean-valued indicator functions (\jl{f_Γ_N} \& \jl{f_Γ_D}) for defining the boundaries $\Gamma_N$ and $\Gamma_D$.

In lines 16--19 we specify the parameters for the finite difference update. The time step coefficient or Courant–Friedrichs–Lewy (CFL) number for the Hamilton-Jacobi evolution equation and reinitialisation equation are given by \jl{γ} and \jl{γ_reinit}, respectively. Line 18 specifies the number of time steps \jl{max_steps} for solving the Hamilton-Jacobi evolution equation. Finally we specify a maximum absolute difference stopping criterion for the reinitialisation equation with tolerance \jl{tol}. It should be noted that we scale \jl{max_steps} and \jl{tol} by the mesh size so that the same set of parameters can be used for problems with different mesh resolutions. 

We specify the problem and output I/O (in/out) parameters in lines 21--26. First, in line 21--23 we define the diffusivity \jl{κ} in $\Omega$, heat flux \jl{g} on $\Gamma_N$, and the required volume fraction \jl{vf}, respectively. The initial level set function is specified in line 24 by \jl{lsf_func}. This must be a \jl{Function} that takes a vector input and returns a scalar. We provide a convenience function \jl{initial_lsf(\ξ,a)} that returns a level set defined by
\begin{equation}
    \boldsymbol{x}\mapsto-\frac{1}{4} \prod_{i=1}^d(\cos(\xi\pi x_i)) - \frac{a}{4}
\end{equation}
where $x_i$ are the components of $\boldsymbol{x}$ and $d$ is the dimension.
Finally, in lines 25--27 we create a \jl{path} for file output and control how often visualisation files are saved using \jl{iter_mod}.

\subsubsection{Finite element setup (lines 28--43)}
In the following we setup the finite element problem using Gridap\footnote{Additional discussion/tutorials for describing finite element problems using Gridap are available at \url{https://gridap.github.io/Tutorials/stable/}.}. This provides a wide range of element types of $n$th order finite elements of various conformity. GridapTopOpt is currently implemented for $n$th order finite elements over rectangular Cartesian meshes.

We use \jl{CartesianDiscreteModel} in line 29 to create a mesh over the computational domain $D$ defined by \jl{dom} and \jl{el_size}. The function \jl{update_labels!} is then used to label parts of the mesh related to $\Gamma_D$ and $\Gamma_N$ using the indicator functions \jl{f_Γ_D} and \jl{f_Γ_D}, respectively. It should be noted that currently GridapTopOpt only supports uniform Cartesian background meshes for solving the Hamilton-Jacobi evolution equation and reinitialisation equation. In future, we plan to remove this requirement using unfitted finite element methods.

Once the model/mesh is defined we create a triangulation and quadrature for both $\Omega$ and $\Gamma_N$. These are created in lines 33--36 where \jl{2*order} indicates the desired Gaussian quadrature degree for numerical integration.

The final stage of the finite element setup is creating the discretised finite element spaces in lines 38--43. We first define a scalar-valued Lagrangian reference element. This is then used to define the test space \jl{V} and trial space \jl{U} corresponding to the discrete approximation of $H^1_{\Gamma_{D}}(\Omega)$. We then construct an FE space \jl{V_φ} over which the level set function is defined, along with an FE test space \jl{V_reg} and trial space \jl{U_reg} over which derivatives are defined. We require that \jl{V_reg} and \jl{U_reg} have zero Dirichlet boundary conditions over regions where we require the extended shape sensitivity to be zero. In general, we do this over regions where a Neumann boundary condition is being applied.

\subsubsection{Level set function (lines 44--47)}
In line 45 we create an FE function for the initial level set by interpolating \jl{lsf_func} onto the finite element space \jl{V_φ}. Figure \ref{fig:thermal2d-results-a} shows a visualisation of this initial structure. In addition, given a smoothing radius $\eta\ll1$ and ersatz material parameter $\epsilon\ll1$, we create an instance of \jl{SmoothErsatzMaterialInterpolation}. This structure provides a smoothed Heaviside function $H_\eta$ (Eqn. \ref{eqn: smoothed heavi}), its derivative $H^\prime_\eta$, the ersatz material approximation $I$ (Eqn. \ref{eqn: erstaz mat}), and the density function $\rho$. In all examples we take the smoothing radius $\eta$ to be twice the maximum side length of an element. 

\subsubsection{Optimisation problem (lines 48--59)}
Gridap provides an expressive API that allows the user to write a weak formulation with a syntax that is near one-to-one with the mathematical notation. This is made possible because Julia natively supports unicode characters\footnote{We refer readers to {\url{https://docs.julialang.org/en/v1/manual/noteworthy-differences/}} for a discussion of important differences between Julia and other programming languages, including MATLAB.} and Gridap turns some of these into Julia functions. For example, the unicode symbol \jl{∇} is an alias for the \jl{gradient} function provided by Gridap.

In line 49 and 50 we define the bi-linear and linear form, respectively for the state equations (Eqn. \ref{eqn: thermal weak form}). We utilise the ersatz material approximation via the term \jl{(I ∘ φ)*κ} where the \jl{∘} operator composes the function \jl{I} with the level set function \jl{φ}. Notice that our weak form almost exactly matches the notation in Eq. \ref{eqn: thermal weak form}. In line 51 we create an instance of \jl{AffineFEStateMap}. This structure is designed to solve the state equations and also provide the necessary ingredients to solve the adjoint problem (Eqn. \ref{eqn: adjoint eqn}) for automatic differentiation if required. By default we use a standard $LU$ solver for computing the solution to the state equations and the adjoint. Later we will see how we can adjust this to use iterative solvers from PETSc. It should be noted that in finite element codes, the biggest memory footprint is often the storage of the sparse matrix for the linear system arising from the discretisation. As a result, the use of automatic differentiation roughly doubles the memory cost for simulations. This is because an additional sparse matrix containing the adjoint system must be stored.

In GridapTopOpt we utilise Gridap's API to also define objective and constraint functionals. In this example we write the objective, constraint and shape derivatives in lines 53--57. Notice that we have slightly rewritten the constraint $C$ so that it is expressed inside the integral over $D$. We then express the shape derivatives given in Eq. \ref{eqn: shape derive thermal comp} and \ref{eqn: shape deriv vol} using the relaxation per Eq. \ref{eqn: lsf integration2}. In line 58 we combine all of these terms in an instance of \jl{PDEConstrainedFunctionals} that handles the  objective and constraints, and their analytic or automatic differentiation. In this case the analytic shape derivatives are passed as optional arguments to the constructor for \jl{PDEConstrainedFunctionals}. When these are not given, automatic differentiation in \jl{φ} is used as discussed in Section \ref{sec: ad impl}. This structure provides the \jl{evaluate!} function that computes the solution to the state equations and derivatives. In addition, the function \jl{get_state} can be used to return the current solution to the finite element problem. 

It is important to note that the measures (\jl{dΩ} and \jl{dΓ_N} in this example) must be included as arguments at the end of the weak formulation, objective, and constraints. This ensures compatibility with Gridap's automatic differentiation in parallel. It should also be noted that we do not currently support the multiplication and division of optimisation functionals. We plan to add this in a future release of GridapTopOpt. 

\subsubsection{Velocity extension (lines 60--63)}
In line 61--63 we define the velocity extension method based on the Hilbertian extension-regularisation problem defined by Eq. \ref{eqn: hilb extension wf}. We use the inner product defined by Eq. \ref{eqn: H1 inner product} with the coefficient $\alpha$ taken to be
\begin{jlcode}
α = 4max_steps*γ*maximum(get_el_Δ(model))
\end{jlcode}
This ensures that as the mesh is refined and \jl{max_steps} increased, the number of elements over which we regularise the gradient is increased. We solve the resulting system using an $LU$ solver by default. This object provides a method \jl{project!} that applies the Hilbertian velocity-extension method to a given shape derivative. 

\subsubsection{Finite difference scheme (lines 64--66)}
To solve the Hamilton-Jacobi evolution equation and reinitialisation equation we use a first order Godunov upwind difference scheme \citep{978-0-387-22746-7_2006,10.1006/jcph.1999.6345_1999}. We first create an instance of \jl{FirstOrderStencil} that represents a finite difference stencil for a single step of the Hamilton-Jacobi evolution equation and reinitialisation equation. We then create an instance of \jl{HamiltonJacobiEvolution} that enables finite differencing over the $n$th order reference finite element in serial or parallel. This object provides the methods \jl{evolve!} and \jl{reinit!} for solving the Hamilton-Jacobi evolution equation and reinitialisation equation, respectively.

Note that there are several methods available for constructing/reinitialising a signed distance function from a level set function. The reader is referred to \cite{978-0-387-22746-7_2006} and \cite{10.1016/bs.hna.2020.10.004_978-0-444-64305-6_2021} and the references therein for a detailed discussion. 

\subsubsection{Optimiser and solution (lines 67--82)}
In GridapTopOpt we implement optimisation algorithms as iterators\footnote{See \url{https://docs.julialang.org/en/v1/manual/interfaces/} for further information.} that inherit from an abstract type \jl{Optimiser}. A concrete \jl{Optimiser} implementation, say \jl{OptEg}, then implements \jl{iterate(m::OptEg)}$\mapsto$\jl{(var,state)} and \jl{iterate(m::OptEg,state)}$\mapsto$\jl{(var,state)}, where \jl{var} and \jl{state} are the available items in the outer loop and internal state of the iterator, respectively. As a result we can iterate over the object \jl{m=OptEg(...)} using \jl{for var in m}. The benefit of this implementation is that the internals of the optimisation method can be hidden in the source code while the explicit \jl{for} loop is still visible to the user. The body of the loop can then be used for auxiliary operations such as writing the optimiser history and other files.

In line 68 and 69 we create an instance of the \jl{AugmentedLagrangian} optimiser based on the discussion by \cite{978-0-387-30303-1_2006}. We slightly modify the method to include a check for oscillations of the Lagrangian using the \jl{has_oscillations} function. If oscillations are detected we reduce the CFL number \jl{γ} for the Hamilton-Jacobi evolution equation by 25\%. In addition, the default stopping criterion at an iteration $q>5$ requires that
\begin{equation}
    \frac{\lvert \mathcal{L}(\Omega^{q})-\mathcal{L}(\Omega^{q-j})\rvert}{\lvert \mathcal{L}(\Omega^{q})\rvert}< \frac{0.01\epsilon_{m}}{d-1},~\forall j = 1,\dots,5
\end{equation}
and
\begin{equation}
    \lvert C_i(\Omega^{q})\rvert < 0.01,~\forall i=1,\dots,N.
\end{equation}
where $\epsilon_{m}$ and $d$ are the maximum side length of an element and the dimension, respectively. Other defaults, such as the maximum iteration number, can also be adjusted as optional arguments (see \jl{@doc AugmentedLagrangian}).

In lines 71--75 we solve the optimisation problem and output VTU files for visualisation and a text file for the optimiser history. Lines 77--79 output a VTU for the final structure in case \jl{iszero(it \% iter_mod) != 0} on line 73 at the final iteration. Figure \ref{fig:thermal2d-results-b} shows the optimisation results for this problem using the visualisation software ParaView \citep{paraview}.

\begin{figure*}[t]
    \centering
    \begin{subfigure}{0.325\textwidth}
        \centering
        \includegraphics[width=\textwidth]{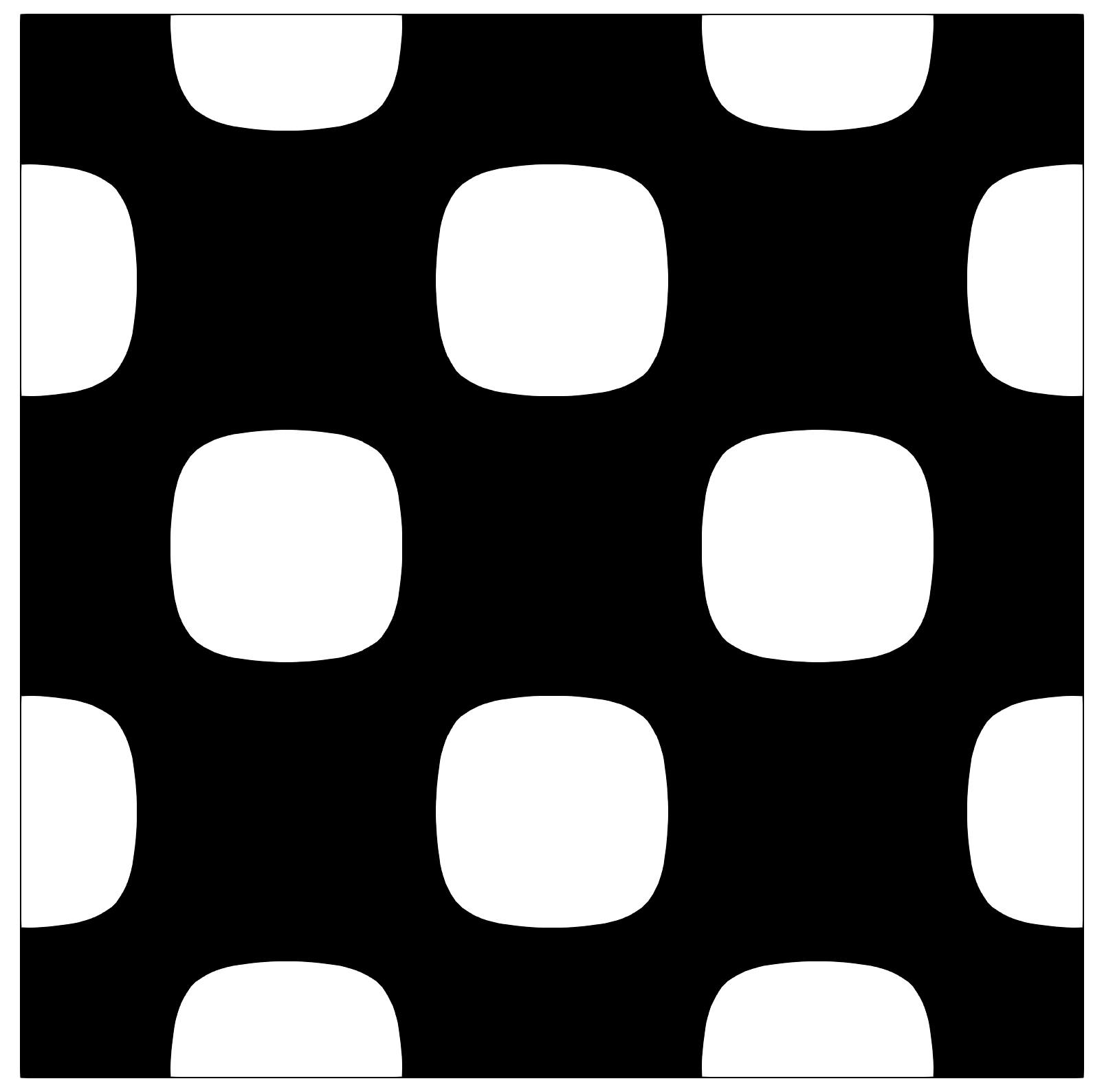}
        \caption{}
        \label{fig:thermal2d-results-a}    
    \end{subfigure}
    \begin{subfigure}{0.325\textwidth}
        \centering
        \includegraphics[width=\textwidth]{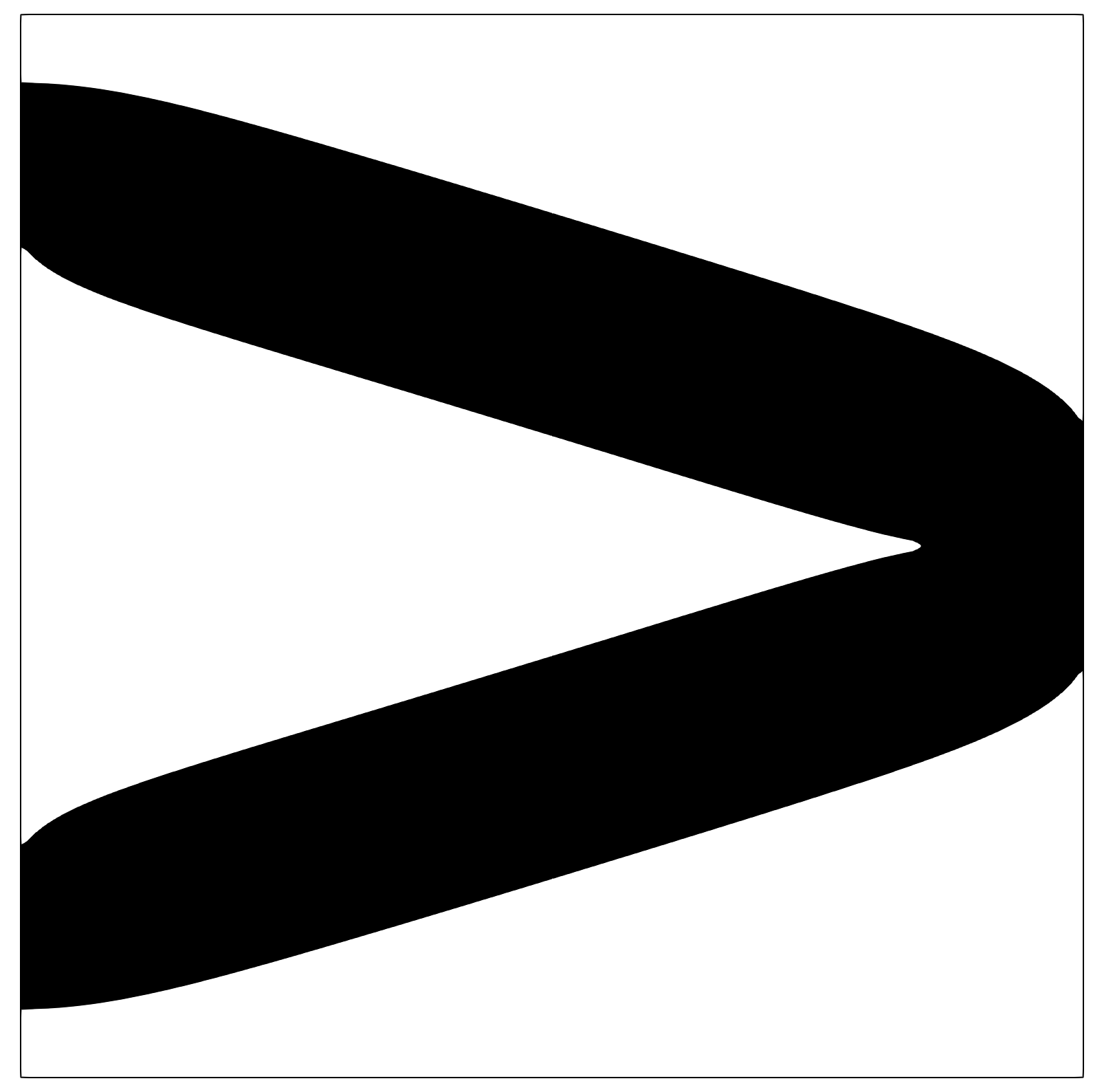}
        \caption{}
        \label{fig:thermal2d-results-b}    
    \end{subfigure}
    \begin{subfigure}{0.325\textwidth}
        \centering
        \includegraphics[width=\textwidth]{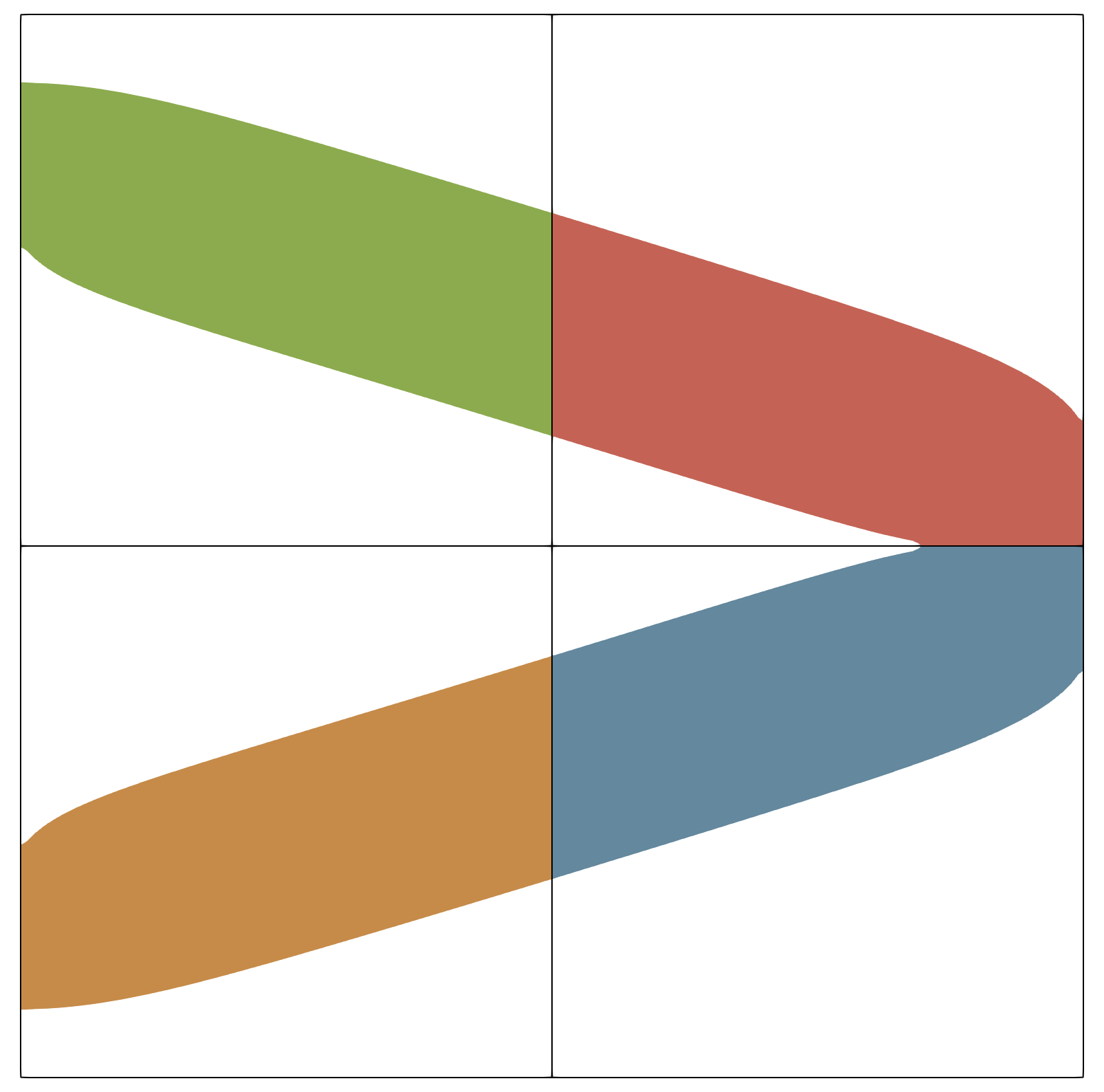}
        \caption{}
        \label{fig:thermal2d-results-c}
    \end{subfigure}
    \caption{Results for the two-dimensional minimum thermal compliance problem at 40\% volume fraction. We visualise the result using an isovolume for $\phi\leq0$. (a) shows the initial structure. (b) and (c) show the final result at convergence in serial and when distributed over a (2,2)-partition, respectively. The colours in (c) show the four partitions of the domain.}
    \label{fig:thermal2d-results}
\end{figure*}

\subsubsection{Extension: MPI \& PETSc}\label{sec: mpi petsc ext}

We now consider some minor modifications to the driver in Appendix \ref{sec: code appendix} to utilise solvers and preconditioners provided by PETSc in parallel. In GridapTopOpt we rely on the GridapPETSc \citep{GridapPETSc_Verdugo2021} satellite package to efficiently interface with the linear and nonlinear solvers provided by the PETSc library. We also utilise the SparseMatricesCSR package \citep{SparseMatrixCSR} because PETSc is based on the \jl{SparseMatrixCSR} datatype. In addition, as discussed previously for distributed computing, we utilise GridapDistributed \citep{Badia2022} and PartitionedArrays \citep{PartitionedArrays_Verdugo2021}. It should be noted that GridapPETSc can also be utilised in serial. To call these packages we change line 1 of Appendix \ref{sec: code appendix} to 
\begin{jlcode}
using GridapTopOpt, Gridap, GridapPETSc, GridapDistributed, PartitionedArrays, SparseMatricesCSR
\end{jlcode}

Before we make any changes to the body of the function \jl{main}, we adjust the end of the driver script to safely launch MPI and PETSc using Julia \jl{do} blocks. We replace line 82 with

\begin{jlcode}
with_mpi() do distribute
  mesh_partition = (2,2)
  solver_options = "-pc_type gamg -ksp_type cg -ksp_error_if_not_converged true 
    -ksp_converged_reason -ksp_rtol 1.0e-12"
  write_dir = ARGS[1]
  GridapPETSc.with(args=split(solver_options)) do
    main(mesh_partition,distribute,write_dir)
  end
end
\end{jlcode}
This requests a (2,2)-partition of the mesh and calls PETSc with the options described in \jl{solver_options}. The options correspond to a conjugate gradient (CG) method with an algebraic multigrid preconditioner (AMG). In addition, we specify the \jl{write_dir} as a command line argument. We then adjust \jl{main} to take two arguments: \jl{mesh_partition} and \jl{distribute}. These are used to create the MPI ranks by replacing line 3 with
\begin{jlcode}
function main(mesh_partition,distribute,write_dir)
  ranks = distribute(LinearIndices((prod(mesh_partition), )))
\end{jlcode}

Next we adjust lines 26--29 as follows:
\begin{jlcode}
path = "$write_dir/therm_comp_MPI/"
i_am_main(ranks) && mkpath(path)
# Model
model = CartesianDiscreteModel(ranks,mesh_partition,dom, el_size);
\end{jlcode}
The function \jl{i_am_main} returns true only on the lead processor and is used here to ensure that the \jl{mkpath} file operations are safely run only once. In addition, we create a partitioned mesh using GridapDistributed by passing \jl{ranks} and \jl{mesh_partition} to \jl{CartesianDiscreteModel}.

To utilise the PETSc linear solver defined by \jl{solver_options} we replace line 51 and 63 with
\begin{jlcode}
Tm = SparseMatrixCSR{0,PetscScalar,PetscInt}
Tv = Vector{PetscScalar}
solver = PETScLinearSolver()
assem_U = SparseMatrixAssembler(Tm,Tv,U,V)
assem_adjoint = SparseMatrixAssembler(Tm,Tv,V,U)
assem_deriv = SparseMatrixAssembler(Tm,Tv,U_reg,U_reg)
state_map = AffineFEStateMap(
  a,l,U,V,V_φ,U_reg,φh,dΩ,dΓ_N;
  assem_U,assem_adjoin,assem_deriv,
  ls = solver,adjoint_ls = solver)
\end{jlcode}
and
\begin{jlcode}
vel_ext = VelocityExtension(
  a_hilb, U_reg, V_reg;
  assem = assem_deriv, ls = solver())
\end{jlcode}
respectively. In the above we specify that the \jl{SparseMatrixAssembler} should assemble \jl{SparseMatrixCSR} matrices with the \jl{PetscScalar} and \jl{PetscInt} datatypes. We then use the solver provided by PETSc for the state equations, adjoint problem (when required), and velocity extension problem using the \jl{PETScLinearSolver} wrapper. To ensure printing is only done on a single processor we adjust line 68 and 69 to be
\begin{jlcode}
optimiser = AugmentedLagrangian(pcfs,stencil,vel_ext,φh; γ,γ_reinit,verbose=i_am_main(ranks))
\end{jlcode}
Finally, we safely write the optimiser history by passing \jl{ranks} to \jl{write_history} in line 74:
\begin{jlcode}
write_history(path*"/history.txt",get_history(optimiser); ranks)
\end{jlcode}

The driver script \jl{therm_comp.jl} (available in \hyperref[rep results]{the source code}) can now be run with 4 MPI tasks and PETSc as:
\begin{jlcode}
mpiexecjl -n 4 julia therm_comp_MPI.jl results/
\end{jlcode}
In Figure \ref{fig:thermal2d-results-c} we show the optimisation result for the distributed problem. As expected it exactly matches the serial result up to the relative tolerance for the iterative solver.

\subsubsection{Extension: 3D}\label{sec: thermal 3d ext}

In this section we consider extending the driver from Section \ref{sec: mpi petsc ext} to handle three-dimensional problems. First, we extend the mesh and MPI partition to three dimensions by setting \jl{el_size = (150,150,150)} and \jl{mesh_partition = (4,6,6)}. This yields approximately 27500 degrees of freedom per processor. In Section \ref{sec: benchmarks}, we will assess the resulting speedup from this choice. Next, we update the indicator functions for $\Gamma_N$ and $\Gamma_D$ to match the diagram in Figure \ref{fig:thermal3d-results-a}:
\begin{jlcode}
f_Γ_N(x) = (x[1] ≈ xmax) && (ymax/2-ymax*prop_Γ_N/2 - eps() <= x[2] <= ymax/2+ymax*prop_Γ_N/2 + eps()) && (zmax/2-zmax*prop_Γ_N/2 - eps() <= x[3] <= zmax/2+zmax*prop_Γ_N/2 + eps())
f_Γ_D(x) = (x[1] ≈ 0.0) && (x[2] <= ymax*prop_Γ_D + eps() || x[2] >= ymax-ymax*prop_Γ_D - eps()) && (x[3] <= zmax*prop_Γ_D + eps() || x[3] >= zmax-zmax*prop_Γ_D - eps())
\end{jlcode}
Finally, for the user parameters, we double the number of \jl{max_steps} for solving the Hamilton-Jacobi evolution equation as we have found that this yields better convergence for three-dimensional problems.

This driver script \jl{therm_comp_3d.jl} (available in \hyperref[rep results]{the source code}) can now be run with 144 MPI tasks as:
\begin{jlcode}
mpiexecjl -n 144 julia therm_comp_MPI_3D.jl results/
\end{jlcode}
To solve this problem we utilise 144 Intel® Xeon® Platinum 8274 Processors with 4GB per core on the Gadi@NCI Australian supercomputer. This optimisation problem takes approximately 15 minutes to solve over 167 iterations. In Figure \ref{fig:thermal3d-results} we show the optimisation results for this problem. Notice that we obtain an analogous result to the two-dimensional problem, where heat is able to follow the shortest path from the heat source to the heat sinks. We note that we achieve a three-dimensional distributed implementation with only the addition of 25 lines including line breaks.

\begin{figure*}[t]
    \centering
        \begin{subfigure}{0.325\textwidth}
            \centering
            \def\svgwidth{\textwidth}
            \large
            \input{Figure4a}
            \caption{}
            \label{fig:thermal3d-results-a} 
    \end{subfigure}
    \begin{subfigure}{0.325\textwidth}
        \centering
        \includegraphics[width=\textwidth]{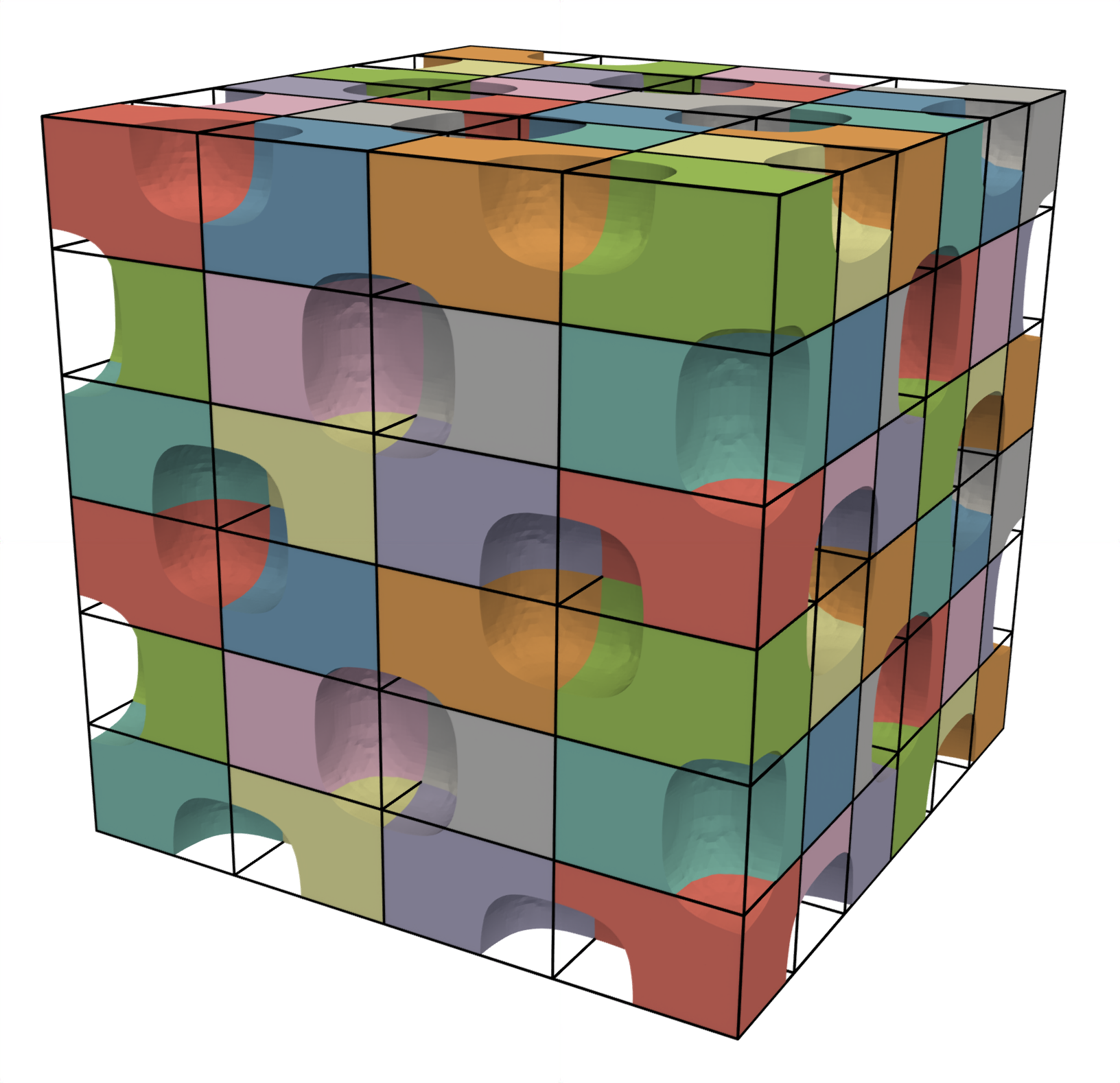}
        \caption{}
        \label{fig:thermal3d-results-b}    
    \end{subfigure}
    \begin{subfigure}{0.325\textwidth}
        \centering
        \includegraphics[width=\textwidth]{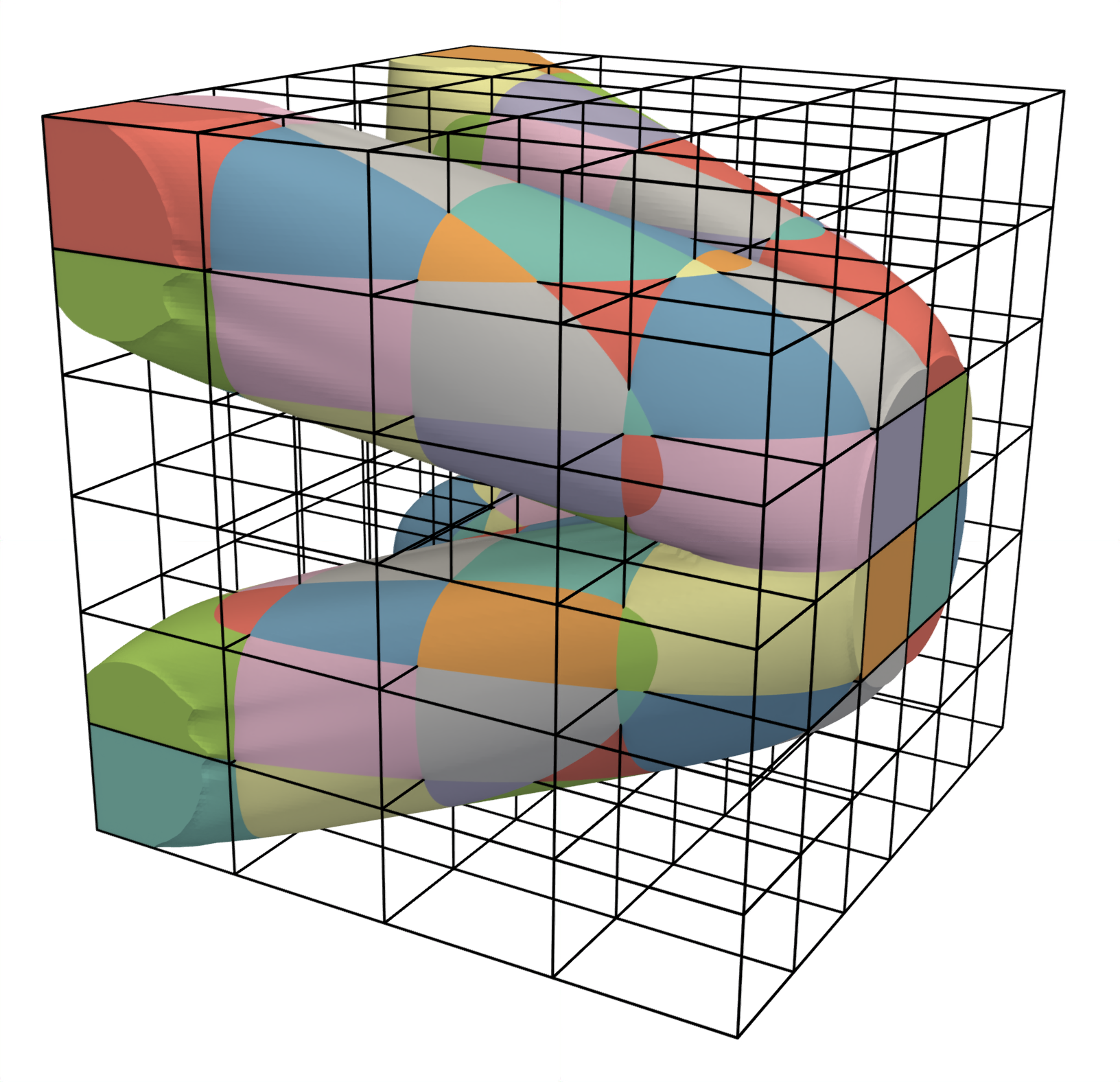}
        \caption{}
        \label{fig:thermal3d-results-c}    
    \end{subfigure}
    \caption{The three-dimensional minimum thermal compliance problem at 40\% volume fraction solved in parallel over a (4,6,6)-partition with a mesh size of $150\times150\times150$. (a) shows the computational domain. (b) and (c) visualise the initial and final structures over the partitions as an isovolume for $\phi\leq0$, respectively.}
    \label{fig:thermal3d-results}
\end{figure*}

\subsection{Minimum elastic compliance}\label{sec: elast comp}
In our next example, we consider the classical three-dimensional minimum compliance problem for a linear elastic symmetric cantilever (see \cite{10.1016/bs.hna.2020.10.004_978-0-444-64305-6_2021}for further details). Figure \ref{fig:cantilever-setup} shows the boundary conditions for this problem. 
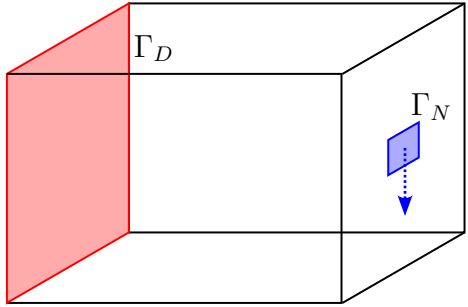
\begin{figure}[t]
    \centering
    \def\svgwidth{0.8\columnwidth}
    \large
    \input{Figure5}
    \caption{A schematic for the three-dimensional cantilever minimum elastic compliance problem. The red boundary denotes a Dirichlet boundary condition on $\Gamma_D$, while the blue arrow denotes an applied stress over $\Gamma_N$.}
    \label{fig:cantilever-setup}
\end{figure}
The corresponding weak formulation of the linear elastic symmetric cantilever is: \textit{Find} $\boldsymbol{u}\in U\coloneqq [H^1_{\Gamma_D}(\Omega)]^3$ \textit{such that}
\begin{equation}\label{eqn: elast weak form}
\int_{\Omega}\boldsymbol{C\varepsilon(u)}:\boldsymbol{\varepsilon(v)}~\mathrm{d}\boldsymbol{x} = \int_{\Gamma_N}\boldsymbol{g}\cdot\boldsymbol{v}~\mathrm{d}s,~\forall \boldsymbol{v}\in U,
\end{equation}
\textit{where} $\boldsymbol{C}$ \textit{is the stiffness tensor,} $\boldsymbol{\varepsilon}$ \textit{is the strain operator, and the operation} $:$ \textit{denotes double contraction}. The optimisation problem is then
\begin{equation}
    \begin{aligned}
    \min_{\Omega\in{D}}&~J(\Omega)\coloneqq\int_{\Omega}\boldsymbol{C\varepsilon(u)}:\boldsymbol{\varepsilon(u)}~\mathrm{d}\boldsymbol{x}\\
    \text{s.t. }&~C(\Omega)=0,\\
    &~a(\boldsymbol{u},\boldsymbol{v})=l(\boldsymbol{v}),~\forall \boldsymbol{v}\in U,
    \end{aligned}
\end{equation}
where we take $D=[0,2]\times[0,1]^2$, $C(\Omega)$ is a volume constraint as previously, and the final line denotes the weak formulation in Eq. \ref{eqn: elast weak form}. Note that the shape derivative of $J(\Omega)$ is given by \cite{10.1016/j.jcp.2003.09.032_2004} as
\begin{equation}\label{eqn: shape derive elast comp}
J'(\Omega)(\boldsymbol{\theta}) = -\int_{\Gamma}\boldsymbol{C\varepsilon(u)}:\boldsymbol{\varepsilon(u)}~\boldsymbol{\theta}\cdot\boldsymbol{n}~\mathrm{d}s,
\end{equation}
where we have again assumed that $\boldsymbol{\theta}\cdot\boldsymbol{n}=0$ on $\Gamma_N$ and $\Gamma_D$.

Below we highlight the important changes to the driver script from Section \ref{sec: thermal 3d ext} to solve the above optimisation problem. First, we adjust the bounding domain to have end points \jl{xmax,ymax,zmax=(2.0,1.0,1.0)}. Next, we generate the stiffness tensor $\boldsymbol{C}$ and create the applied load via \jl{C = isotropic_elast_tensor(3,1.0,0.3)} and \jl{g = VectorValue(0,0,-1)} where the first argument of \jl{isotropic_elast_tensor} is the dimension of the problem, and the latter arguments are the Young's modulus and Poisson's ratio of the base material, respectively. To update the boundary $\Gamma_D$ to match Figure \ref{fig:cantilever-setup} we adjust the indicator function to be \jl{f_Γ_D(x) = (x[1] ≈ 0.0)}. Next, we adjust the reference element and function spaces to handle vector-valued solutions:
\begin{jlcode}
reffe = ReferenceFE(lagrangian,VectorValue{3,Float64}, order)
reffe_scalar = ReferenceFE(lagrangian,Float64,order)
V = TestFESpace(model,reffe;dirichlet_tags=["Gamma_D"])
U = TrialFESpace(V,VectorValue(0.0,0.0,0.0))
V_φ = TestFESpace(model,reffe_scalar)
V_reg = TestFESpace(model,reffe_scalar;dirichlet_tags= ["Gamma_N"])
U_reg = TrialFESpace(V_reg,0)
\end{jlcode}
The weak form and optimisation functional $J$ are then written as
\begin{jlcode}
a(u,v,φ,dΩ,dΓ_N) = ∫((I ∘ φ)*(C ⊙ ε(u) ⊙ ε(v)))dΩ
l(v,φ,dΩ,dΓ_N) = ∫(v⋅g)dΓ_N
\end{jlcode}
and
\begin{jlcode}
J(u,φ,dΩ,dΓ_N) = ∫((I ∘ φ)*(C ⊙ ε(u) ⊙ ε(u)))dΩ
dJ(q,u,φ,dΩ,dΓ_N) = ∫((C ⊙ ε(u) ⊙ ε(u))*q*(DH ∘ φ)*(norm ∘ ∇(φ)))dΩ
\end{jlcode}
respectively, where the operator \jl{⊙} denotes double contraction of tensors in Gridap. Finally, to solve the state equations we use a custom GridapPETSc solver \jl{ElasticitySolver} that implements CG-AMG for elastic-type problems, and for the Hilbertian extension-regularisation problem we use a standard CG-AMG solver via \jl{PETScLinearSolver} as previously. 

We now solve this problem over a mesh of size $160\times80\times80$ via a (4,6,6)-partition using 144 MPI tasks as:
\begin{jlcode}
mpiexecjl -n 144 julia elast_comp_MPI_3D.jl results/
\end{jlcode}
This driver script is available in the \hyperref[rep results]{the source code}. In Figure \ref{fig:cantilever-results-b} we show the results for this optimisation problem given the initial structure in Figure \ref{fig:cantilever-results-a}. This optimisation problem takes approximately 38 minutes to solve over 108 iterations. 
\begin{figure*}[t]
    \centering
    \begin{subfigure}{0.325\textwidth}
        \centering
        \includegraphics[width=\textwidth]{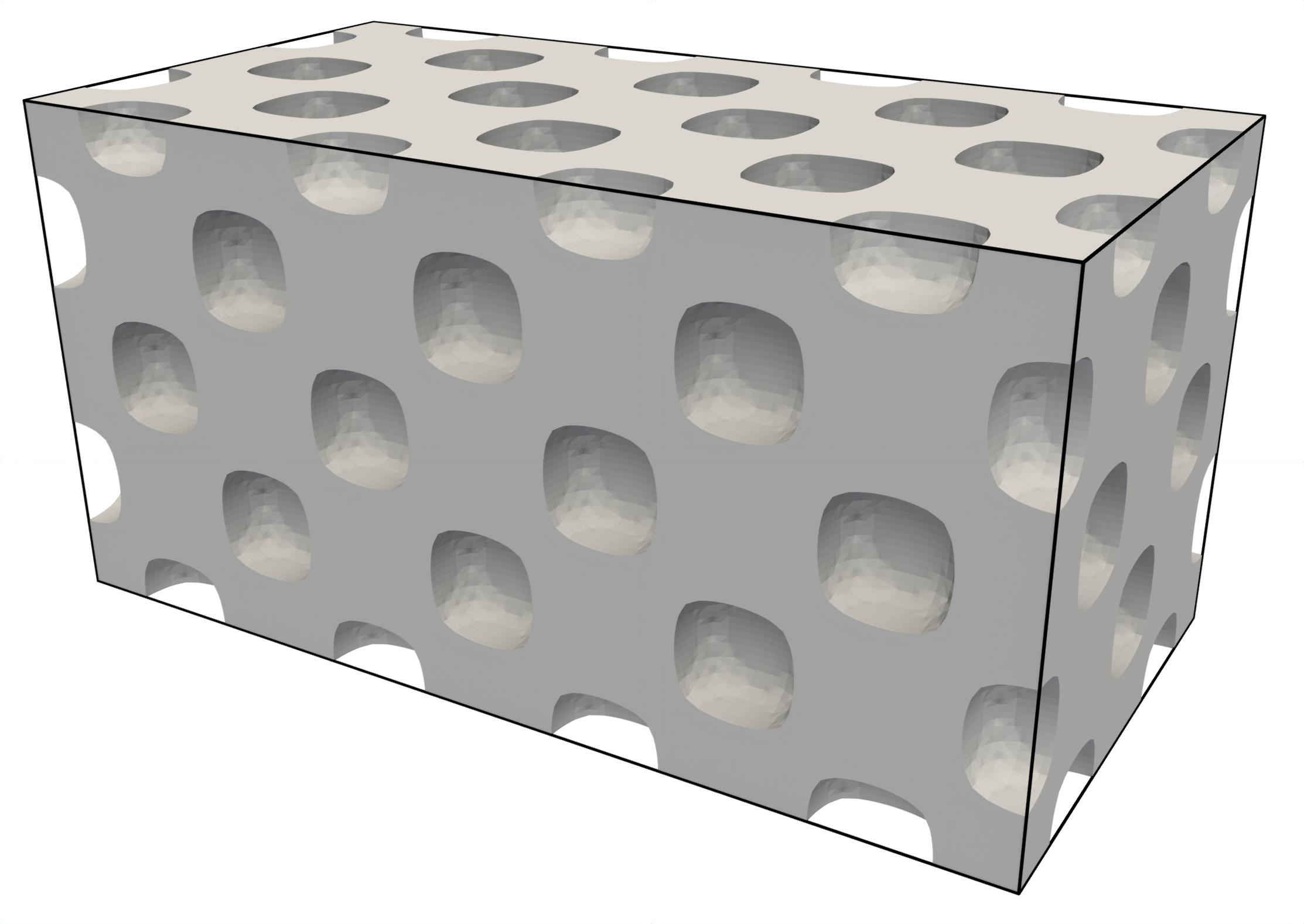}
        \caption{}
        \label{fig:cantilever-results-a}    
    \end{subfigure}
    \begin{subfigure}{0.325\textwidth}
        \centering
        \includegraphics[width=\textwidth]{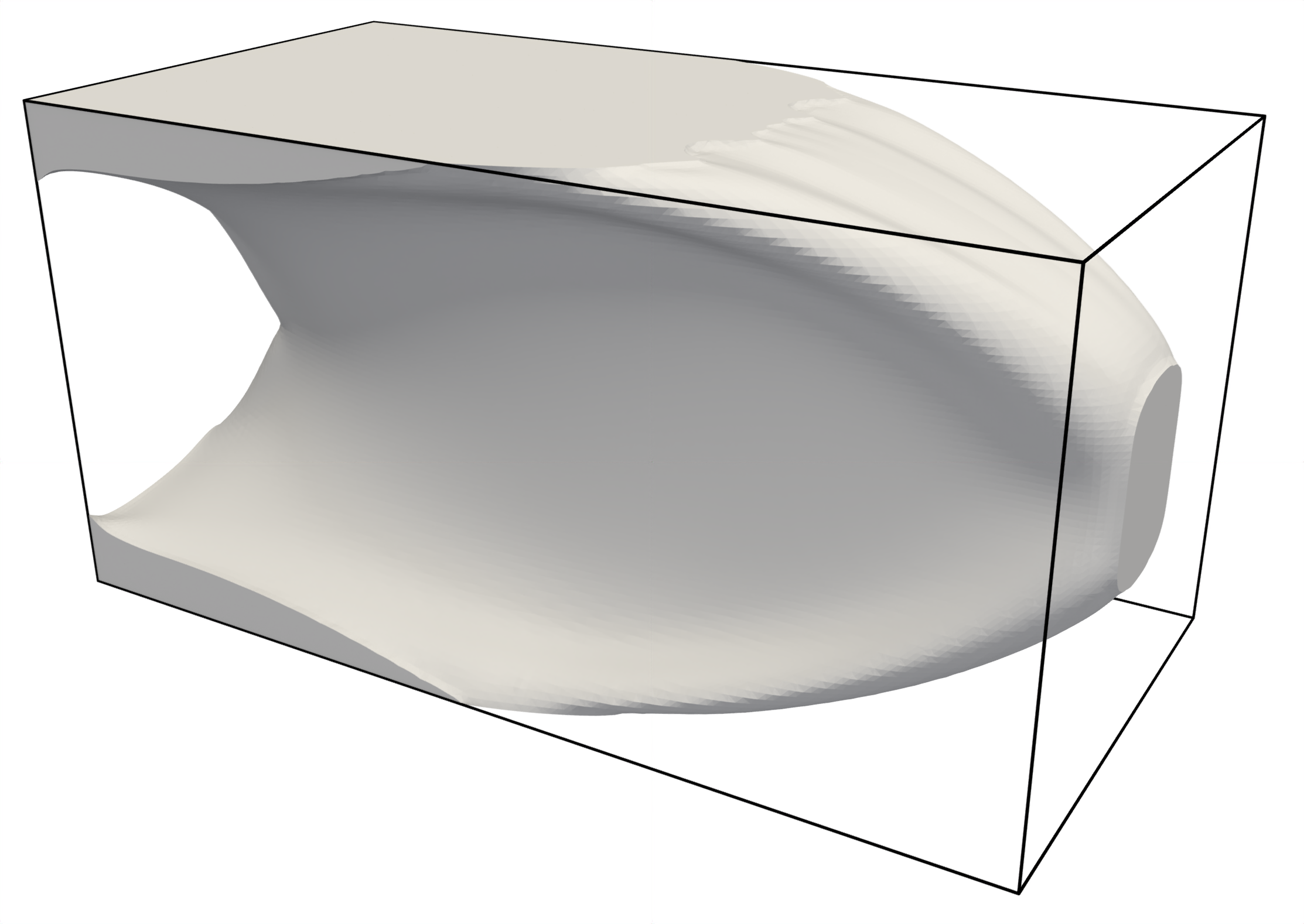}
        \caption{}
        \label{fig:cantilever-results-b}    
    \end{subfigure}
    \begin{subfigure}{0.325\textwidth}
        \centering
        \includegraphics[width=\textwidth]{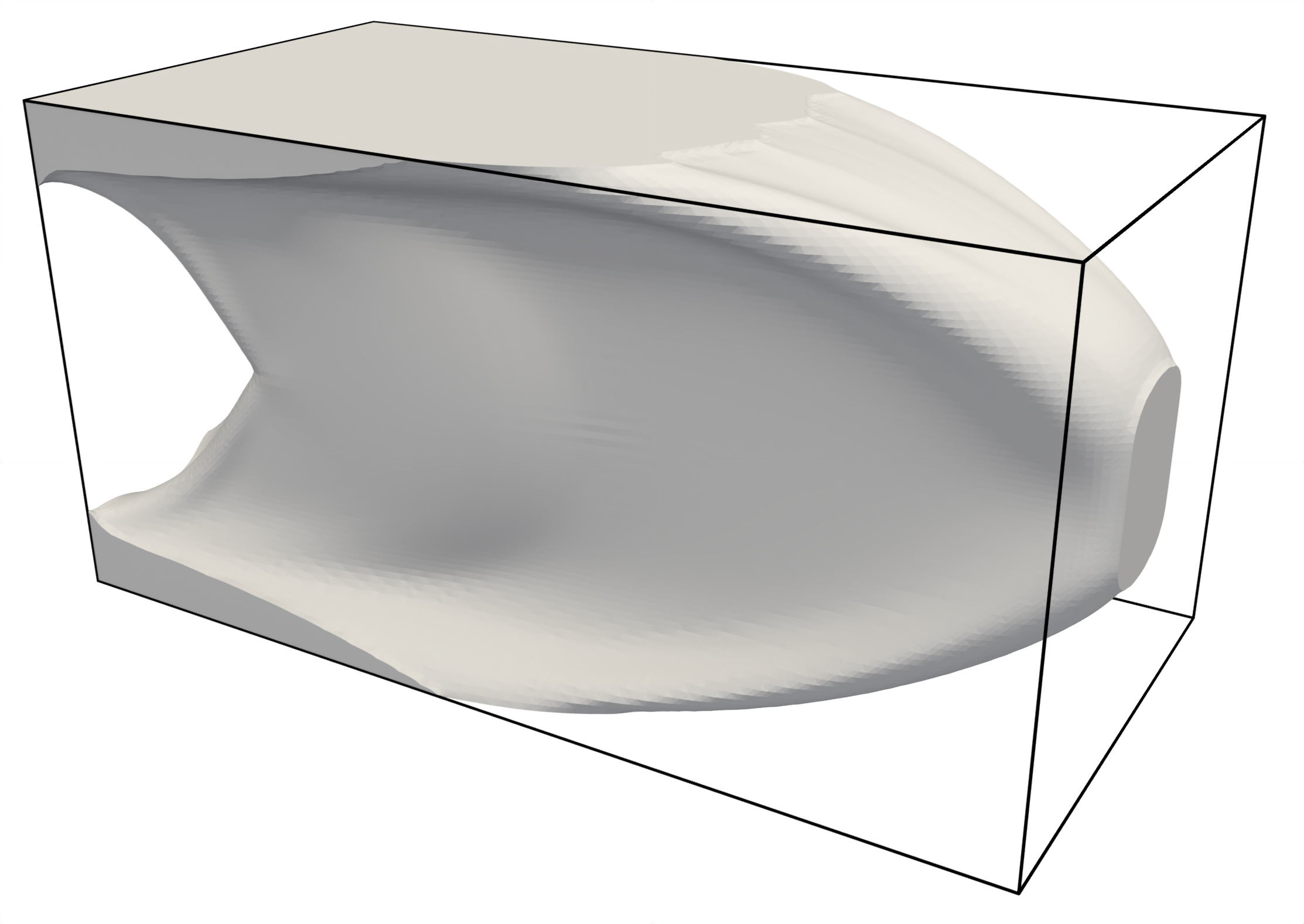}
        \caption{}
        \label{fig:cantilever-results-c}    
    \end{subfigure}
    \caption{Results for the three-dimensional cantilever minimum elastic/hyperelastic compliance problem at 40\% volume fraction. We visualise the result using an isovolume for $\phi\leq0$. (a) shows the initial structure while (b) and (c) show the linear elastic and hyperelastic optimisation result, respectively.}
    \label{fig:cantilever-results}
\end{figure*}

\subsubsection{Extension: hyperelasticity}\label{sec: hyperelast comp}
GridapTopOpt has the capability to solve problems with non-linear weak formulations. In this extension we consider the previous symmetric cantilever problem with a hyperelastic consitutitive law for a Neo-Hookean solid \citep{Bonet_Wood_2008}. The corresponding weak formulation of this problem is: \textit{Find} $\boldsymbol{u}\in U\coloneqq [H^1_{\Gamma_D}(\Omega)]^3$ \textit{such that}
$\mathcal{R}(\boldsymbol{u},\boldsymbol{v})=0$, $\forall \boldsymbol{v}\in U$\textit{, where}
\begin{equation}
\mathcal{R}(\boldsymbol{u},\boldsymbol{v})=\int_{\Omega}\boldsymbol{S}(\boldsymbol{u}):\mathrm{d}{\boldsymbol{E}}(\boldsymbol{u},\boldsymbol{v})~\mathrm{d}\boldsymbol{x} - \int_{\Gamma_N}\boldsymbol{g}\cdot\boldsymbol{v}~\mathrm{d}s,
\end{equation}
$\boldsymbol{S}$ \textit{is the second Piola-Kirchhoff tensor}
\begin{equation}
    \boldsymbol{S}(\boldsymbol{u}) = \mu(\boldsymbol{I}-\boldsymbol{C}^{-1})+\lambda\ln(J)\boldsymbol{C}^{-1},
\end{equation}
\textit{and} $\mathrm{d}{\boldsymbol{E}}$ \textit{is a variation of the Green strain tensor} $\boldsymbol{E}=\frac{1}{2}\left(\boldsymbol{F}^\intercal\boldsymbol{F}-\boldsymbol{I}\right)$ \textit{given by}
\begin{equation}
    \mathrm{d}{\boldsymbol{E}}(\boldsymbol{u},\boldsymbol{v}) = \frac{1}{2}\left(\boldsymbol{F}\cdot\boldsymbol{\nabla}(\boldsymbol{v})+(\boldsymbol{\nabla}(\boldsymbol{v})\cdot\boldsymbol{F})^\intercal\right).
\end{equation}
In the above, $\mu$ and $\lambda$ are the Lam\'e parameters, $\boldsymbol{C}=\boldsymbol{F}^\intercal\boldsymbol{F}$ is the right Cauchy-Green deformation tensor, $\boldsymbol{F}=\boldsymbol{I}+\boldsymbol{\nabla u}$ is the deformation gradient tensor, and $J=\sqrt{\det\boldsymbol{C}}$ captures the volume change. Further mathematical details for this problem are described by \cite{Bonet_Wood_2008}. Our optimisation problem is then 
\begin{equation}
    \begin{aligned}
    \min_{\Omega\in{D}}&~J(\Omega)\coloneqq\int_{\Omega}\boldsymbol{S}(\boldsymbol{u}):\mathrm{d}{\boldsymbol{E}}(\boldsymbol{u},\boldsymbol{u})~\mathrm{d}\boldsymbol{x}\\
    \text{s.t. }&~C(\Omega)=0,\\
    &~\mathcal{R}(\boldsymbol{u},\boldsymbol{v})=0,~\forall \boldsymbol{v}\in U,
    \end{aligned}
\end{equation}
where $C(\Omega)$ and $D$ are as previously for the linear problem.

Owing to near one-to-one correspondence between mathematical notation and syntax in Gridap and GridapTopOpt, extending the previous script to handle non-linearity of the weak form is straightforward. The most important modification is switching from an \jl{AffineFEStateMap} to a \jl{NonlinearFEStateMap} to handle the nonlinear state equations and enable automatic differentiation via the adjoint method. This is implemented via:
\begin{jlcode}
lin_solver = ElasticitySolver(V)
nl_solver = NewtonSolver(lin_solver;maxiter=50,rtol=10^-8, verbose=i_am_main(ranks))
state_map = NonlinearFEStateMap(
  res,U,V,V_φ,U_reg,φh,dΩ,dΓ_N;
  assem_U = SparseMatrixAssembler(Tm,Tv,U,V),
  assem_adjoint = SparseMatrixAssembler(Tm,Tv,V,U),
  assem_deriv = SparseMatrixAssembler(Tm,Tv,U_reg,U_reg),
  nls = nl_solver, adjoint_ls = lin_solver)
\end{jlcode}
We use a Newton–Raphson method to solve the nonlinear weak formulation and utilise the previous \jl{ElasticitySolver} for the linear problem inside Newton-Raphson and for the adjoint problem. We then utilise automatic differentiation for the objective functional by omitting the optional \jl{analytic_dJ} from \jl{PDEConstrainedFunctionals}. We reuse memory between non-linear iterations of the Newton-Raphson method. As a result, the memory cost of the nonlinear problem with automatic differentiation is comparable to that of a linear problem with automatic differentiation. It should be noted that we do not focus on the development of scalable solvers for the resolution of this nonlinear problem as this is still an open area of research. In addition, this implementation is suitable for small to moderate strains only owing to the ersatz material approximation. In addition, the exploration of efficient nonlinear solvers is outside the scope of this paper. Both of these are on-going areas of research in topology optimisation \citep[e.g.,][]{Chen_Wang_Wang_Zhang_2017,Ortigosa_Martínez-Frutos_Gil_Herrero-Pérez_2019}. In future, we hope to explore solutions to these restrictions using provably-robust solvers and unfitted finite element methods.

We now solve this problem over a mesh of size $160\times80\times80$ via a (4,6,6)-partition using 144 MPI tasks as:
\begin{jlcode}
mpiexecjl -n 144 julia hyperelast_comp_MPI_3D.jl results/
\end{jlcode}
This driver script is available in the \hyperref[rep results]{the source code}. 
In Figure \ref{fig:cantilever-results-c} we show the result for this optimisation problem. As expected, one may observe asymmetry of the structure compared to the linear result. This optimisation problem takes approximately 8 hours to solve over 139 iterations. 

\subsection{Inverter mechanism}

In this example we utilise the automatic differentiation capabilities of GridapTopOpt to solve a three-dimensional inverter mechanism problem in parallel. Figure \ref{fig:cantilever-setup} shows the boundary conditions for this problem, that are similar to \cite{10.1007/s00158-018-1950-2_2018}.
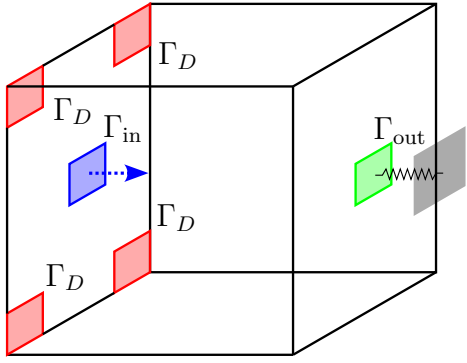
\begin{figure}[t]
    \centering
    \def\svgwidth{0.8\columnwidth}
    \large
    \input{Figure7}
    \caption{A schematic for the three-dimensional inverter mechanism problem. The red boundary denotes a Dirichlet boundary condition on $\Gamma_D$, the blue arrow denotes an applied stress over $\Gamma_{\text{in}}$, and the green boundary denotes a robin boundary condition on $\Gamma_{\text{out}}$}
    \label{fig:inverter-setup}
\end{figure}
The corresponding weak formulation for the linear elastic problem is: {\textit{Find} $\boldsymbol{u}\in U\coloneqq [H^1_{\Gamma_D}(\Omega)]^3$ \textit{such that}
\begin{align}\label{eqn: inverter form}
\int_{\Omega}\boldsymbol{C\varepsilon(u)}&:\boldsymbol{\varepsilon(v)}~\mathrm{d}\boldsymbol{x} + \int_{\Gamma_{\text{out}}}k_s\boldsymbol{u}\cdot\boldsymbol{v}~\mathrm{d}\Gamma\nonumber\\&= \int_{\Gamma_N}\boldsymbol{g}\cdot\boldsymbol{v}~\mathrm{d}s,~\forall \boldsymbol{v}\in U,
\end{align}
\textit{where} $k_s$ \textit{is an artificial spring stiffness}.} The optimisation problem that we consider here is minimising the average positive $x$ displacement on $\Gamma_{\text{in}}$, while constraining the volume and the average negative $x$ displacement on $\Gamma_{\text{out}}$. Mathematically, this problem is given by:
\begin{equation}
    \begin{aligned}
    \min_{\Omega\in{D}}&~J(\Omega)\coloneqq\frac{1}{\operatorname{Vol}(\Gamma_{\text{in}})}\int_{\Gamma_{in}}\boldsymbol{u}\cdot\boldsymbol{e}_1~\mathrm{d}\Gamma\\
    \text{s.t. }&~C_1(\Omega)=0,\\
    &~C_2(\Omega)=0,\\
    &~a(\boldsymbol{u},\boldsymbol{v})=l(\boldsymbol{v}),~\forall \boldsymbol{v}\in U,
    \end{aligned}
\end{equation}
where $\boldsymbol{e}_{1}=(1,0,0)$, $C_1$ is the volume constraint as previously, and $C_2$ is given by 
\begin{equation}
    C_2(\Omega)=\frac{1}{\operatorname{Vol}(\Gamma_{\text{out}})}\int_{\Gamma_{out}}(\boldsymbol{u}\cdot\boldsymbol{e}_1-\delta_x)\mathrm{d}\Gamma,
\end{equation}
where $\delta_x<0$ is the required average negative displacement.

The modification of the linear elastic cantilever script to handle this problem is straightforward and the driver script is provided as \jl{inverter_3d.jl} in the \hyperref[rep results]{the source code}. It is worth noting that we slightly modify the initial level set function to avoid small islands of material on the face containing $\Gamma_{\text{out}}$. The initial structure is shown in Figure \ref{fig:inverter-results-a}. In addition, we require that a small region near $\Gamma_{\text{out}}$ has zero extended shape sensitivity to ensure that it remains solid over the course of the optimisation history.

We solve this problem over a mesh of size $100\times100\times100$ via a (4,6,6)-partition using 144 MPI tasks as:
\begin{jlcode}
mpiexecjl -n 144 julia inverter_MPI_3D.jl results/
\end{jlcode}
In Figure \ref{fig:inverter-results} we show the optimisation result for this problem and visualise the deformation under an applied load. With no parameter tuning, this optimisation problem takes approximately 5.5 hours to solve over 523 iterations. With further parameter tuning this could be reduced. 

\begin{figure*}[t]
    \centering
    \begin{subfigure}{0.325\textwidth}
        \centering
        \includegraphics[width=\textwidth]{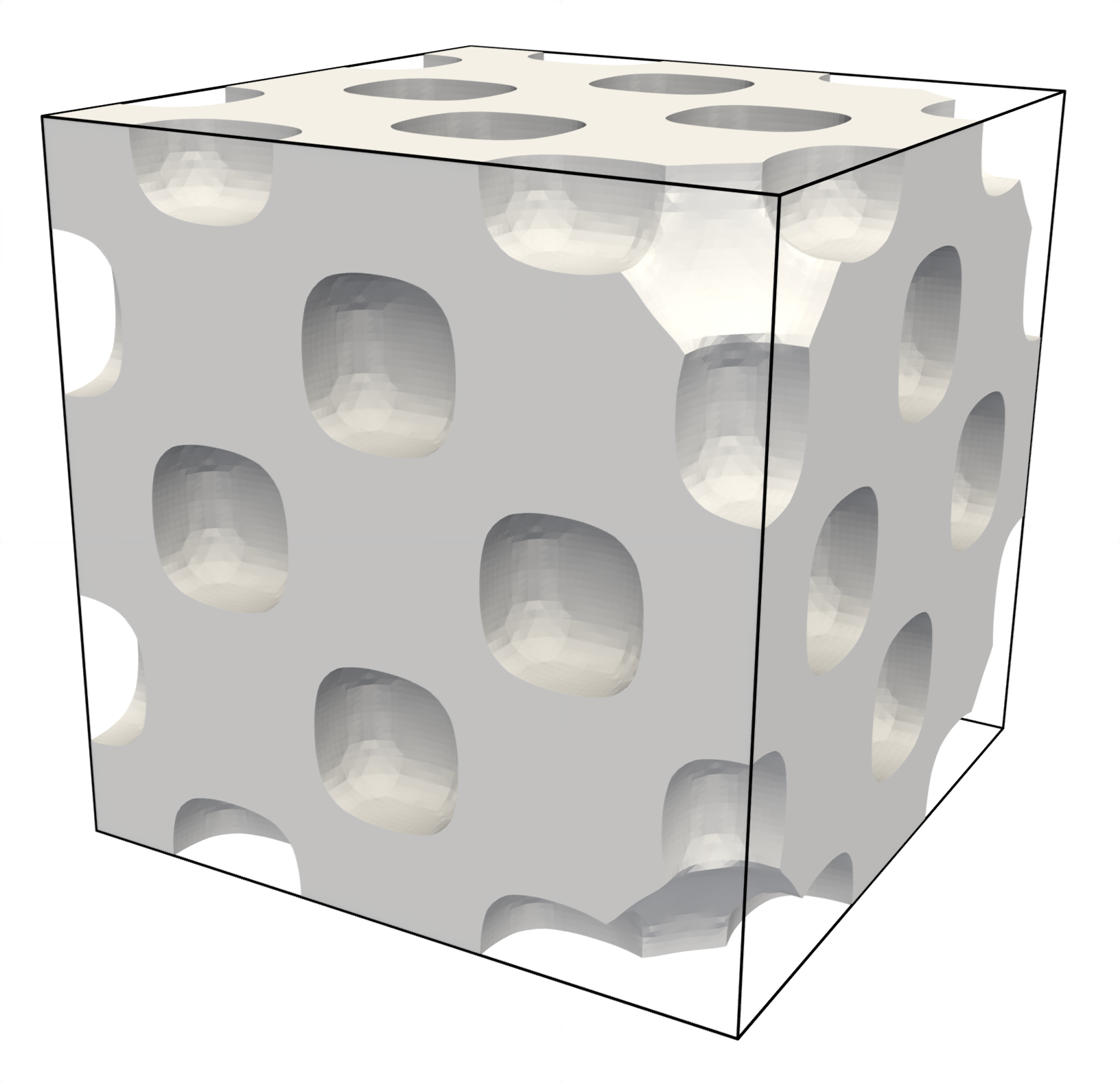}
        \caption{}
        \label{fig:inverter-results-a}    
    \end{subfigure}
    \begin{subfigure}{0.325\textwidth}
        \centering
        \includegraphics[width=\textwidth]{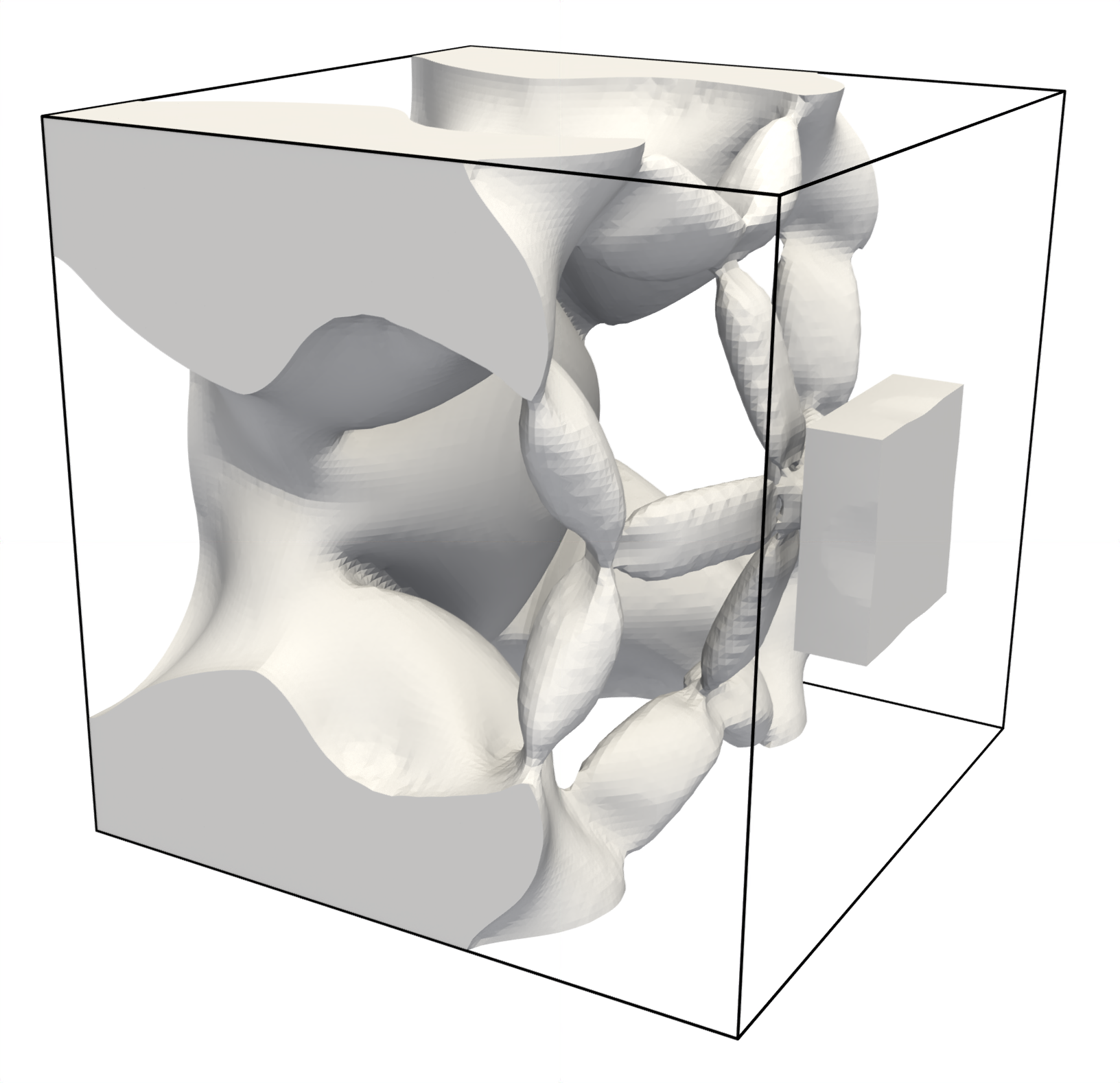}
        \caption{}
        \label{fig:inverter-results-b}    
    \end{subfigure}
    \begin{subfigure}{0.325\textwidth}
        \centering
        \includegraphics[width=\textwidth]{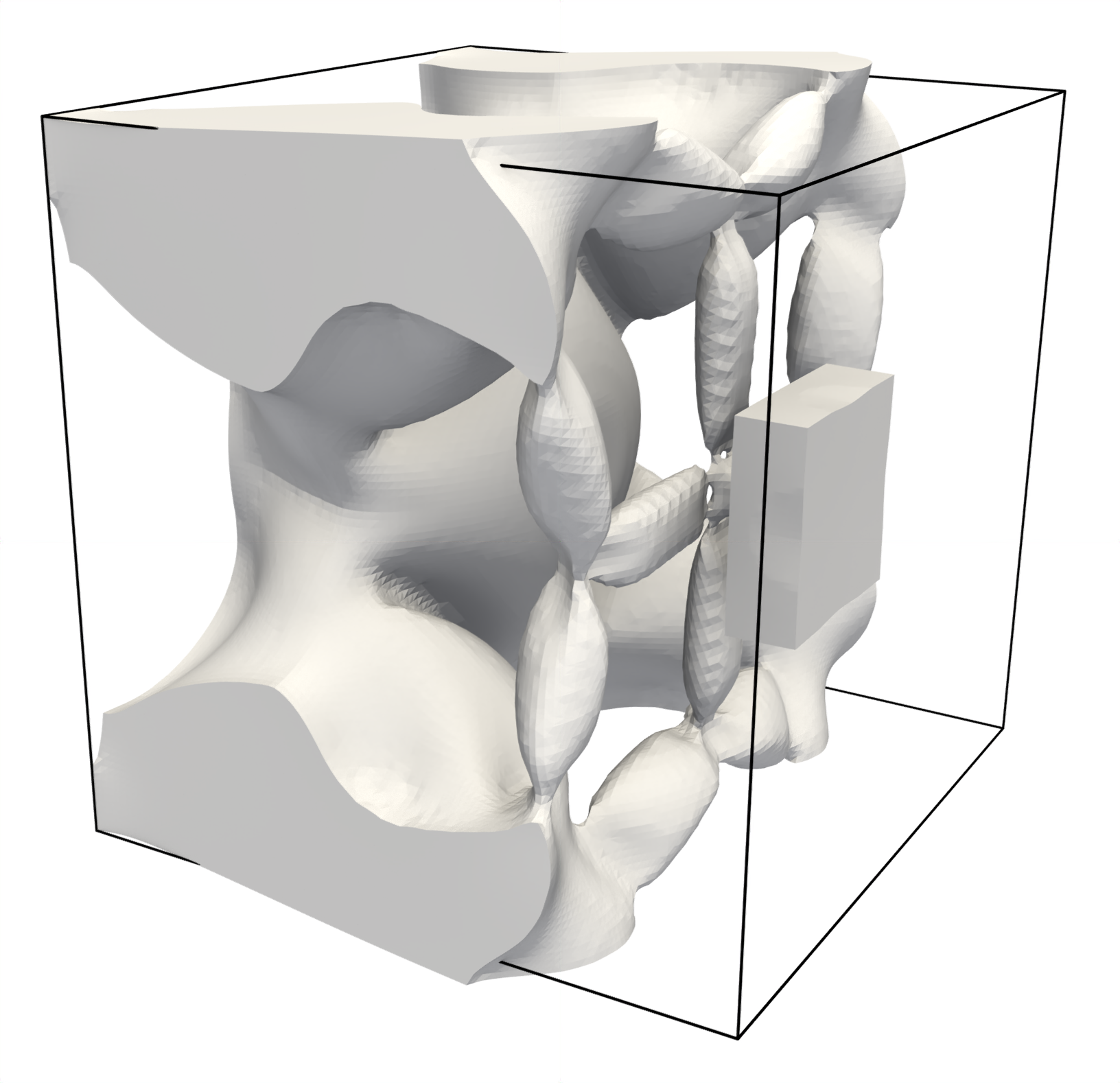}
        \caption{}
        \label{fig:inverter-results-c}    
    \end{subfigure}
    \caption{Results for the three-dimensional inverter mechanism problem at 40\% volume fraction. As previously, we visualise the result using an isovolume for $\phi\leq0$. (a) shows the initial structure, (b) shows the result at convergence, and (c) shows the resulting deformation under load at 25\% scaling.}
    \label{fig:inverter-results}
\end{figure*}

It should be noted that thickness constraints can be considered for this type of problem to avoid material with very thin joints. In the context of level-set methods, thickness constraints can be constructed using properties of the signed distance function \citep{10.1007/s00158-016-1453-y_2016}. As a result, the shape derivatives of these constraints require special treatment \citep{10.1007/s00158-016-1453-y_2016,10.1051/m2an/2019056_2020}. We plan to investigate these types of functionals in the future.

\subsection{Linear elastic inverse homogenisation}

Our final example is a three-dimensional linear elastic inverse homogenisation problem, whereby we seek to find a periodic microstructure that optimises a homogenised macroscopic material property. The corresponding weak formulation of the state equation is: \textit{For each unique constant macroscopic strain field} $\boldsymbol{\bar{\varepsilon}}^{(k l)}$\textit{, find} ${\boldsymbol{u}}^{(kl)}\in H^1_{\textrm{per}}(\Omega)^3$ \textit{such that}
\begin{equation}\label{eqn: le weak form}
    \begin{aligned}
    &\int_{\Omega} \boldsymbol{C} \boldsymbol{\varepsilon}({\boldsymbol{u}}^{(k l)}):\boldsymbol{\varepsilon}(\boldsymbol{v}) ~\mathrm{d} \Omega\\& \hspace{0.5cm} = -\int_{\Omega} \boldsymbol{C}\boldsymbol{\bar{\varepsilon}}^{(kl)}:\boldsymbol{\varepsilon}(\boldsymbol{v}) ~\mathrm{d} \Omega,~\forall\, \boldsymbol{v}\in H^1_{\textrm{per}}(\Omega)^d
    \end{aligned}
\end{equation}
\textit{where} $\bar{\varepsilon}_{i j}^{(k l)}=\frac{1}{2}\left(\delta_{i k} \delta_{j l}+\delta_{i l} \delta_{j k}\right)$ \textit{in index notation and} $H^1_{\text{per}}(D)$ \textit{is the periodic Sobolev space}. There are six unique macroscopic strain fields $\boldsymbol{\bar{\varepsilon}}^{(k l)}$ in three dimensions, and there are six corresponding displacements ${\boldsymbol{u}}^{(kl)}$ at the solution to the above weak formulation. This problem is therefore an example of a multi-field finite element problem, where each bilinear form is identical and repeated for six different linear forms.

We consider maximising the homogenised bulk modulus $\bar{\kappa}$ or resistance to deformation of the macroscopic material, subject to a volume constraint. This optimisation problem is written as
\begin{equation}
    \begin{aligned}
    \min_{\Omega\in{D}}&~-\bar{\kappa}(\Omega)\\
    \text{s.t. }&~C(\Omega)=0,\\
    &~a({\boldsymbol{u}}_p,\boldsymbol{v})=l_p(\boldsymbol{v}),~\forall \boldsymbol{v}\in U, i = 1,\dots,6,
    \end{aligned}
\end{equation}
where we take $D=[0,1]^3$, we use Voigt notation to denote the six unique displacements as ${\boldsymbol{u}}_p \equiv {\boldsymbol{u}}^{(kl)}$ corresponding to each linear form $l_p$, $C(\Omega)$ is the volume constraint as previously, and $\bar{\kappa}$ is the homogenised bulk modulus given by
\begin{align*}
    \bar{\kappa} = \frac{1}{9}(&\bar{C}_{1111}+\bar{C}_{2222}+\bar{C}_{3333}\\&+2(\bar{C}_{1122}+\bar{C}_{1133}+\bar{C}_{2233}))
\end{align*}
where $\bar{C}_{ijkl}$ is the effective stiffness tensor in index notation. This effective stiffness tensor is given by \citep{YvonnetCompHomogenization}
\begin{equation}
    \bar{C}_{ijkl}(\Omega) = \int_\Omega \boldsymbol{C}(\boldsymbol{\varepsilon}({\boldsymbol{u}}^{(ij)})+\boldsymbol{\bar\varepsilon}^{(ij)})\boldsymbol{\bar\varepsilon}^{(kl)}~\mathrm{d}\Omega,
    \label{eqn: le hom}
\end{equation}
with shape derivative \citep{10.1007/s00158-023-03663-0_2023}
\begin{equation}
\begin{aligned}
    &\bar{C}_{ijkl}^{\prime}(\Omega)(\boldsymbol{\theta})= \int_{\partial\Omega} \boldsymbol{C}(\boldsymbol{\varepsilon}({\boldsymbol{u}}^{(ij)})+\boldsymbol{\bar\varepsilon}^{(ij)})\\&\hspace{2cm}\times(\boldsymbol{\varepsilon}({\boldsymbol{u}}^{(kl)})+\boldsymbol{\bar\varepsilon}^{(kl)})~\boldsymbol{\theta}\cdot\boldsymbol{n}~\mathrm{d}{\Gamma}.
\end{aligned}
\end{equation}

The driver script to solve this problem is provided as \jl{inverse_hom_3d.jl} in the \hyperref[rep results]{the source code}. Below we discuss several important adjustments that have been made in this script. First, we demonstrate how to initialise a different level set function by setting \jl{lsf_func} equal to the zero-level set of a Schwarz P minimal surface. Next, we utilise a periodic Cartesian mesh via the optional argument \jl{isperiodic=(true,true,true)} in \jl{CartesianDiscreteModel}. We then use \jl{f_Γ_D(x) = iszero(x)} to identify a single point at the origin that we will use to sufficiently constrain the problem under periodicity. For the bilinear and linear forms we take
\begin{jlcode}
εᴹ = (TensorValue(1.,0.,0.,0.,0.,0.,0.,0.,0.),
      TensorValue(0.,0.,0.,0.,1.,0.,0.,0.,0.),    
      TensorValue(0.,0.,0.,0.,0.,0.,0.,0.,1.),    
      TensorValue(0.,0.,0.,0.,0.,1/2,0.,1/2,0.),  
      TensorValue(0.,0.,1/2,0.,0.,0.,1/2,0.,0.),  
      TensorValue(0.,1/2,0.,1/2,0.,0.,0.,0.,0.))  
a(u,v,φ,dΩ) = ∫((I ∘ φ) * C ⊙ ε(u) ⊙ ε(v))dΩ
l = [(v,φ,dΩ) -> ∫(-(I ∘ φ) * C ⊙ εᴹ[i] ⊙ ε(v))dΩ for i in 1:6]
\end{jlcode}
where the entries in \jl{εᴹ} correspond to each unique macroscopic strain $\boldsymbol{\bar{\varepsilon}}^{(k l)}$. Note that the way we have written the above is slightly atypical when compared to usual multi-field problems\footnote{See example \url{https://gridap.github.io/Tutorials/stable/pages/t007_darcy/}.}. This description allows us to leverage the fact that we have a repeated bilinear form for each corresponding linear form. In particular, because each numerical \textit{block} corresponding to each bilinear form is identical, we avoid allocating the entire stiffness matrix and instead utilise a single block. We do this in a way that also facilitates automatic differentiation. To instantiate this object we replace the \jl{AffineFEStateMap} with a \jl{RepeatedAffineFEStateMap} as
\begin{jlcode}
state_map = RepeatingAffineFEStateMap(
  6,a,l,U,V,V_φ,U_reg,φh,dΩ; ...)
\end{jlcode}
where the first argument denotes the number of repeated blocks down the diagonal in the multi-field problem and \jl{...} denotes the previous assemblers and solvers for elasticity. The first argument should match the number of linear forms from above. In principle this implementation can readily handle multi-physics inverse homogenisation problems such as thermoelastic inverse homogenisation.

The final modification consists of writing the objective:
\begin{jlcode}
Cᴴ(r,s,u,φ,dΩ) = ∫((I ∘ φ)*(C ⊙ (ε(u[r])+εᴹ[r]) ⊙ εᴹ[s]))dΩ
dCᴴ(r,s,q,u,φ,dΩ) = ∫(-q*(C ⊙ (ε(u[r])+εᴹ[r]) ⊙ (ε(u[s])+εᴹ[s]))*(DH ∘ φ)*(norm ∘ ∇(φ)))dΩ
κ(u,φ,dΩ) = -1/9*(Cᴴ(1,1,u,φ,dΩ) + Cᴴ(2,2,u,φ,dΩ) + Cᴴ(3,3,u,φ,dΩ)+2*(Cᴴ(1,2,u,φ,dΩ)+Cᴴ(1,3,u,φ,dΩ) + Cᴴ(2,3,u,φ,dΩ)))
dκ(q,u,φ,dΩ) = -1/9*(dCᴴ(1,1,q,u,φ,dΩ) + dCᴴ(2,2,q,u,φ,dΩ) + dCᴴ(3,3,q,u,φ,dΩ)+2*(dCᴴ(1,2,q,u,φ,dΩ) + dCᴴ(1,3,q,u,φ,dΩ) + dCᴴ(2,3,q,u,φ,dΩ)))
\end{jlcode}
Note that in the above we use \jl{⋅ᴴ} instead of $\bar{\cdot}$ to denote homogenised quantities. In addition, the indexing over \jl{u} retrieves the corresponding single-field solution from the multi-field finite element solution. 

We solve this problem over a mesh of size $100\times100\times100$ via a (4,6,6)-partition using 144 MPI tasks as:
\begin{jlcode}
mpiexecjl -n 144 julia inverse_hom_MPI_3D.jl results/
\end{jlcode}
In Figure \ref{fig:inverse-hom-results} we show the optimisation results for this problem. This optimisation problem takes approximately 4 hours to solve over 356 iterations.

\begin{figure*}[t]
    \centering
    \begin{subfigure}{0.325\textwidth}
        \centering
        \def\svgwidth{\textwidth}
        \large
        \input{Figure9a}
        \caption{}
        \label{fig:inverse-hom-a}   
    \end{subfigure}
    \begin{subfigure}{0.325\textwidth}
        \centering
        \includegraphics[width=\textwidth]{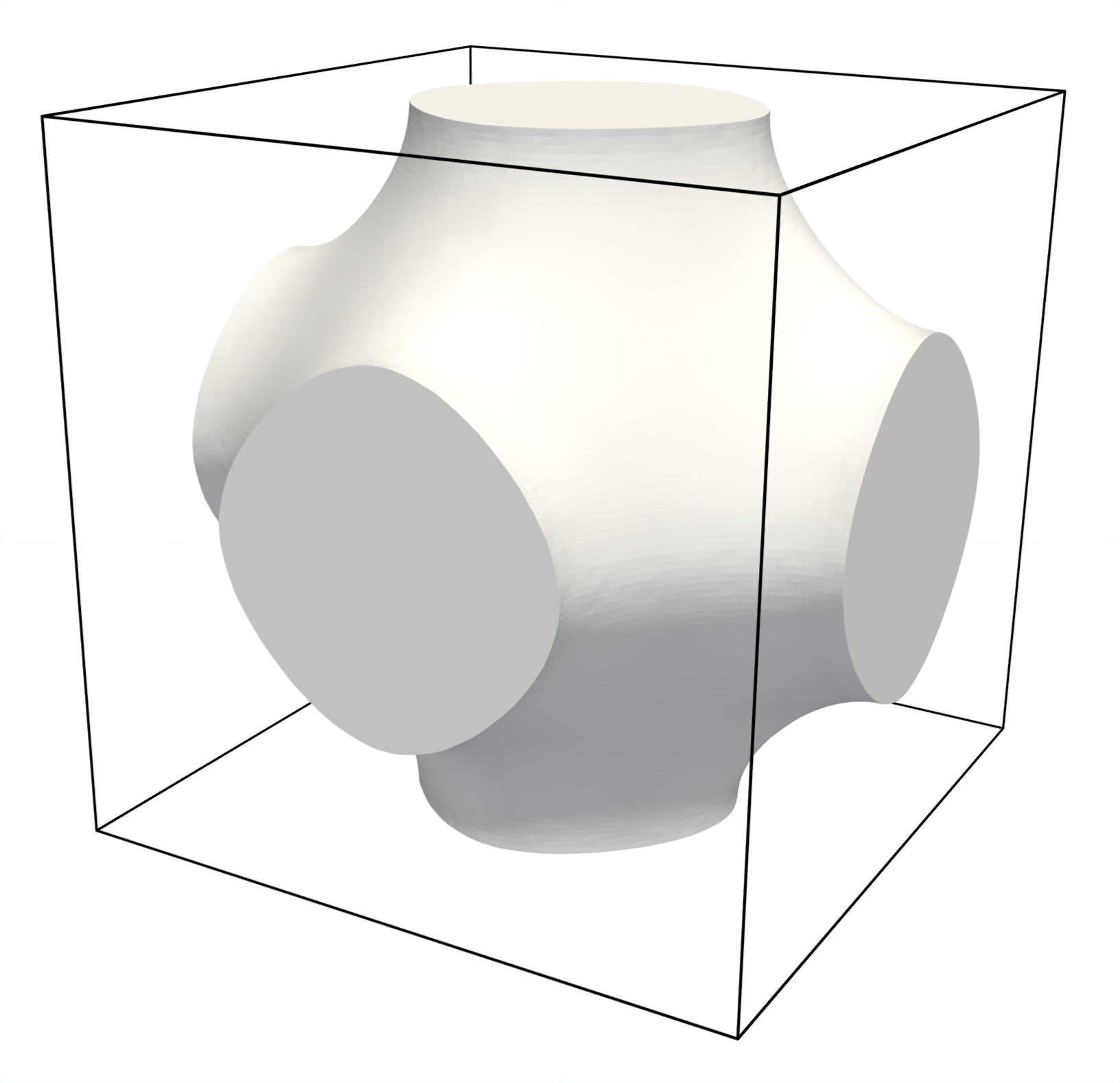}
        \caption{}
        \label{fig:inverse-hom-b}    
    \end{subfigure}
    \begin{subfigure}{0.325\textwidth}
        \centering
        \includegraphics[width=\textwidth]{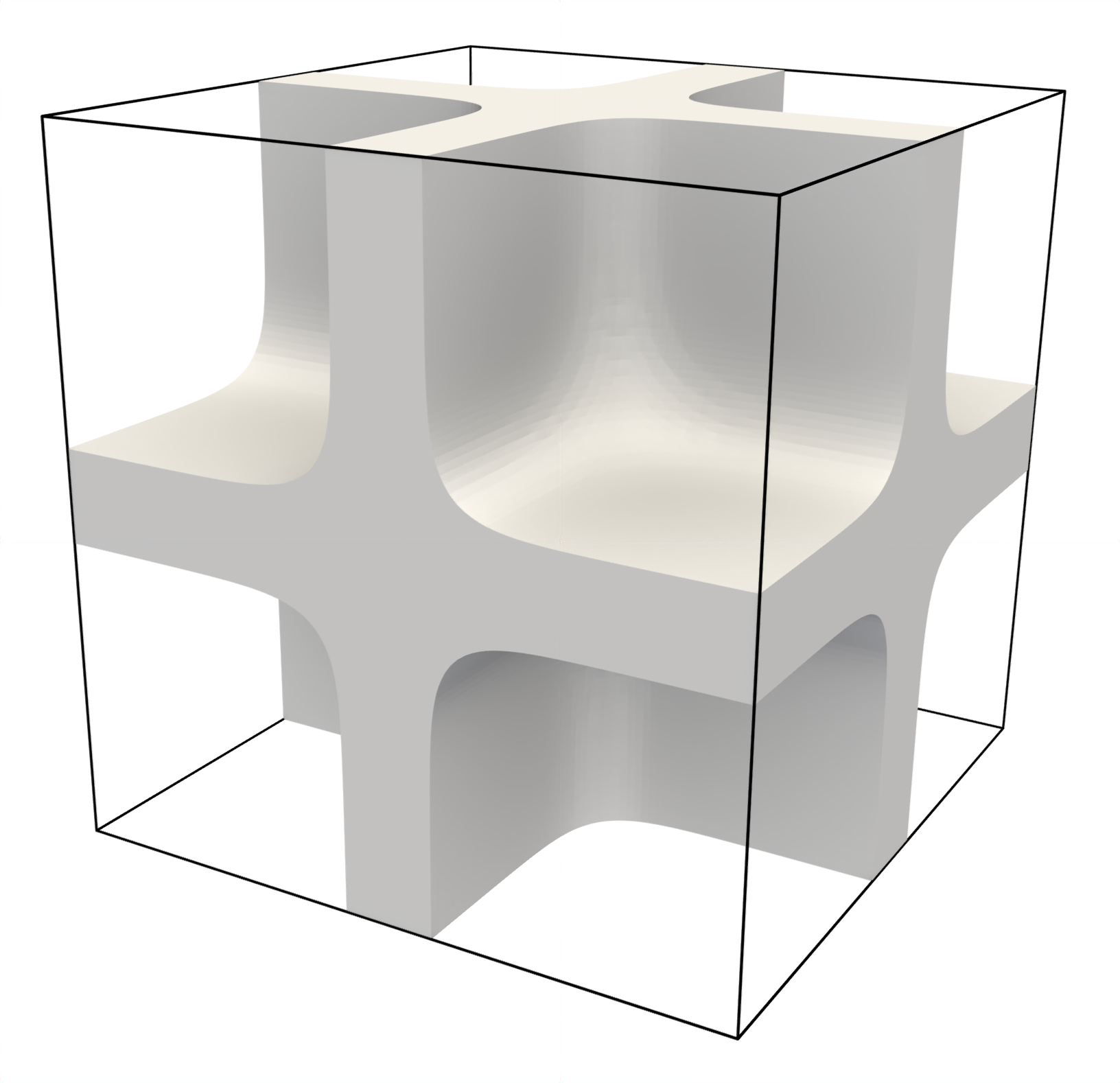}
        \caption{}
        \label{fig:inverse-hom-c}    
    \end{subfigure}
    \caption{Results for the three-dimensional inverse homogenisation problem at 40\% volume fraction. As previously, we visualise the result using an isovolume for $\phi\leq0$. Figure (a) shows the periodic boundary conditions, Figure (b) shows the initial level set function, and Figure (c) shows the result at convergence.}
    \label{fig:inverse-hom-results}
\end{figure*}

\section{Benchmarks}\label{sec: benchmarks}
In this section we benchmark the speedup and scaling capabilities of GridapTopOpt. We consider two types of benchmark:
\begin{itemize}
    \item \textit{Strong scaling}: the \textbf{total} number of degrees of freedom is fixed, while the number of processors is increased.
    \item \textit{Weak scaling}: the number of degrees of freedom \textbf{per processor} is fixed, while the number of processors is increased.
\end{itemize}
The first of these demonstrates the computational speedup obtained by increasing the number of processors for a fixed problem size, while the second illustrates the scaling capability of algorithms as we increase the problem size.

For this paper, we consider benchmarking a single iteration of the three-dimensional thermal problem (Sec. \ref{sec: thermal 3d ext}), elastic problem (Sec. \ref{sec: elast comp}), and hyperelastic problem (Sec. \ref{sec: hyperelast comp}) with $D=[0,1]^3$ for all problems. In particular, we consider the maximum time across all processors for the second iteration of the augmented Lagrangian algorithm for each problem. We then take the minimum time across the 10 independent runs to find the best obtainable performance. For the strong scaling test we fix the number of degrees of freedom at $3\times101^3$ and for the weak scaling test we fix degrees of freedom per processor at 32000. As previously we benchmark using Intel® Xeon® Platinum 8274 Processors with 4GB per core on the Gadi@NCI Australian supercomputer. In our testing we consider 8, 27, 48, 96, 144, 240, 384, 624, 864, and 1104 CPU processors. It should be noted that partitions are assigned to processors automatically via MPI and we do not explore any sort of thread pinning/process binding in our simulations.

\begin{figure*}[t]
    \centering
    \includegraphics[width=1\textwidth]{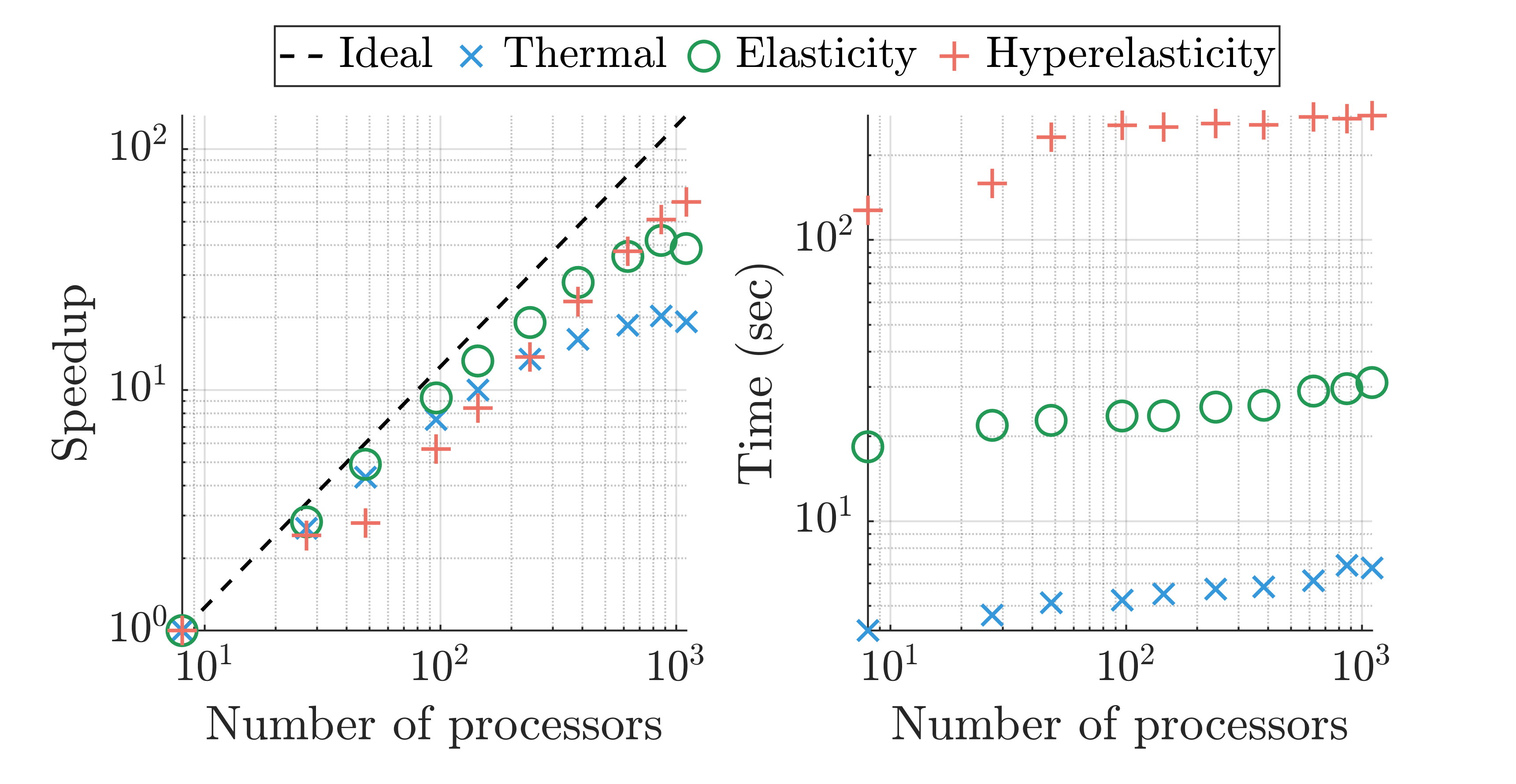}
    \caption{GridapTopOpt single-iteration benchmarks for different state equations. Figure (a) shows a strong scaling benchmark, whereby the number of degrees of freedom is fixed while the number of processors is increased. The dashed line shows the ideal strong scaling where we assume infinite communication speed and zero network latency. Figure (b) shows a weak scaling benchmark where the number of degrees of freedom per processor is fixed, and the number of processors is increased.}
    \label{fig:benchmarks}
\end{figure*}

Figure \ref{fig:benchmarks} shows the strong and weak scaling benchmarks for these problems. The strong scaling benchmark shows that increasing the number of processors increases the speedup until a point where communication overhead sufficiently reduces efficiency. On the other hand, the weak scaling results show some increase in execution time as the number of processors is increased. For the thermal and elasticity problems this is expected because the number of iterations in CG-AMG slightly increases as the number of processors is increased due to the AMG preconditioner not being weakly scalable. For the hyperelastic case, we expect a much larger computation time as we solve several linear systems per iteration. Overall, these results demonstrate reasonable speedup and scaling capabilities.

\section{Conclusions}\label{sec: concl}

In this paper we have presented an open-source computational toolbox for level set-based topology optimisation that is implemented in Julia using the Gridap package ecosystem. The core design principle of GridapTopOpt is to provide an extendable framework for solving PDE-constrained optimisation problems in serial or parallel with a high-level programming interface and automatic differentiation. The package provides users with the tools to efficiently solve a wide array of optimisation problems such as linear/nonlinear problems, inverse homogenisation problems, and multi-physics problems. We have explained the usage and demonstrated the versatility of GridapTopOpt by formulating and solving several topology optimisation problems with distinct underlying PDEs and optimisation functionals. Our strong and weak scaling benchmarks of the parallel implementation demonstrate the speedup and scaling capability of the package. In future we hope to extend the framework to parallel h-adaptive unfitted finite element methods, as well as enhance the scalability of our solvers by leveraging recent and future developments in GridapSolvers \cite{Manyer2024}. These future developments will allow GridapTopOpt to solve an even wider range of problems (e.g., fluid-structure interaction and nonlinear
inverse homogenisation).

\backmatter

\bmhead{Acknowledgments}
This work was supported by the Australian Research Council through the Discovery Projects grant scheme (DP220102759). Preliminary development of the software used computational resources provided by the eResearch Office, Queensland University of Technology. The research used computational resources provided by National Computational Infrastructure (NCI) Australia, a National Collaborative Research Infrastructure Strategy enabled capability supported by the Australian Government. The first author is supported by a QUT Postgraduate Research Award and a Supervisor Top-Up Scholarship. The above support is gratefully acknowledged. The first and last author would also like to thank the School of Mathematics, Monash University for hosting them as part of this project. Finally, we would like to thank the reviewers for their insightful comments that helped improve the manuscript.

\bmhead{Conflict of interest}
The authors have no competing interests to declare that are relevant to the content of this article.

\bmhead{Replication of Results}\label{rep results}
The source code used for this paper and additional computational details have been provided at \url{https://github.com/zjwegert/GridapTopOpt.jl/tree/Wegert_et_al_2024}.

\begin{appendices}

\section{Thermal compliance shape derivative}\label{sec: thermal shape deriv appendix}
In this appendix we show the following result:
\begin{lemma}
The shape derivative of
\begin{equation}
    J(\Omega)=\int_{\Omega}\kappa\boldsymbol{\nabla}u\cdot\boldsymbol{\nabla}u~\mathrm{d}\boldsymbol{x}
\end{equation}
is given by
\begin{equation}
\begin{aligned}
    J^\prime(\Omega)(\boldsymbol{\theta}) = -\int_{\Gamma}\kappa\boldsymbol{\nabla}u\cdot\boldsymbol{\nabla}u~\boldsymbol{\theta}\cdot\boldsymbol{n}~\mathrm{d}s
\end{aligned}
\end{equation}
with $\boldsymbol{\theta}\cdot\boldsymbol{n}=0$ on $\Gamma_N$ and $\Gamma_D$. 
\end{lemma}
\begin{proof}
We proceed via C\'ea's method \citep{10.1051/m2an/1986200303711_1986}. Suppose we define the Lagrangian $\mathcal{L}$ to be
\begin{align*}
    \mathcal{L}(\Omega,v)&=-\int_\Omega (\kappa\boldsymbol{\nabla}v\cdot\boldsymbol{\nabla}v)\,\mathrm{d}\Omega+2\int_{\Gamma_N}gu\,\mathrm{d}\Gamma\\&\quad+2\int_{\Gamma_D}v\kappa\boldsymbol{\nabla}v\cdot\boldsymbol{n}\,\mathrm{d}\Gamma
\end{align*}
We do not include auxiliary fields as it turns out that the problem is self-adjoint. Requiring the Lagrangian to be stationary and taking a partial derivative with respect to $v$ in the direction $\varphi$ gives
\begin{align*}
    0&=\pderiv{\mathcal{L}}{v}(\varphi)=\D{}{\varepsilon}\bigg{\rvert}_{\varepsilon=0}\mathcal{L}(\Omega,v+\varepsilon\varphi)\\
    &=-2\int_\Omega \kappa\boldsymbol{\nabla}v\cdot\boldsymbol{\nabla}\varphi\,\mathrm{d}\Omega+2\int_{\Gamma_N} g \varphi\,\mathrm{d}\Gamma\\&\quad+2\int_{\Gamma_D} v\kappa\boldsymbol{\nabla}\varphi\cdot\boldsymbol{n}+\varphi\kappa\boldsymbol{\nabla}v\cdot\boldsymbol{n}\,\mathrm{d}\Gamma\\
    &=2\int_\Omega \boldsymbol{\nabla}\cdot(\kappa\boldsymbol{\nabla}v)\varphi\,\mathrm{d}\Omega-2\int_{\Gamma} \varphi \kappa\boldsymbol{\nabla}v\cdot\boldsymbol{n}\,\mathrm{d}\Gamma\\&\quad-2\int_{\Gamma_N} \varphi(\kappa\boldsymbol{\nabla}v\cdot\boldsymbol{n}-g)\,\mathrm{d}\Gamma+2\int_{\Gamma_D} v\kappa\boldsymbol{\nabla}\varphi\cdot\boldsymbol{n}\,\mathrm{d}\Gamma
\end{align*}
where we have used integration by parts. As the Lagrangian is stationary we recover Eq. \ref{eqn: heat eqn}. In addition, at the solution $u$ of Eq. \ref{eqn: heat eqn} we obtain $\mathcal{L}(\Omega,u)=J(\Omega)$ as required. Finally, applying the chain rule when shape differentiating $\mathcal{L}$ we obtain
\begin{align*}
    J^\prime(\Omega)(\boldsymbol{\theta})&=\pderiv{\mathcal{L}}{\Omega}(\boldsymbol{\theta}).
\end{align*}
Applying Lemma 4 and 5 \citep{10.1016/j.jcp.2003.09.032_2004} and using the fact that $\Gamma_N$ and $\Gamma_D$ are fixed with $\boldsymbol{\theta}\cdot\boldsymbol{n}=0$, we obtain the result:
\begin{equation*}
    J^\prime(\Omega)(\boldsymbol{\theta}) = -\int_{\Gamma}\kappa\boldsymbol{\nabla}u\cdot\boldsymbol{\nabla}u~\boldsymbol{\theta}\cdot\boldsymbol{n}~\mathrm{d}s
\end{equation*}
\end{proof}

\onecolumn
\section{Thermal compliance code}\label{sec: code appendix}
\begin{minted}{julia}
using GridapTopOpt, Gridap

function main(write_dir)
  # FE parameters
  order = 1                                       # Finite element order
  xmax,ymax = (1.0,1.0)                           # Domain size
  dom = (0,xmax,0,ymax)                           # Bounding domain
  el_size = (200,200)                             # Mesh partition size
  prop_Γ_N = 0.2                                  # Γ_N size parameter
  prop_Γ_D = 0.2                                  # Γ_D size parameter
  f_Γ_N(x) = (x[1] ≈ xmax &&                      # Γ_N indicator function
    ymax/2-ymax*prop_Γ_N/2 - eps() <= x[2] <= ymax/2+ymax*prop_Γ_N/2 + eps())
  f_Γ_D(x) = (x[1] ≈ 0.0 &&                       # Γ_D indicator function
    (x[2] <= ymax*prop_Γ_D + eps() || x[2] >= ymax-ymax*prop_Γ_D - eps()))
  # FD parameters
  γ = 0.1                                         # HJ eqn time step coeff
  γ_reinit = 0.5                                  # Reinit. eqn time step coeff
  max_steps = floor(Int,order*minimum(el_size)/10)# Max steps for advection
  tol = 1/(5order^2)/minimum(el_size)		          # Reinitialisation tolerance
  # Problem parameters
  κ = 1                                           # Diffusivity
  g = 1                                           # Heat flow in
  vf = 0.4                                        # Volume fraction constraint
  lsf_func = initial_lsf(4,0.2)                   # Initial level set function
  iter_mod = 10                                   # VTK Output modulo
  path = "$write_dir/therm_comp_serial/"          # Output path
  mkpath(path)                                    # Create path
  # Model
  model = CartesianDiscreteModel(dom,el_size);
  update_labels!(1,model,f_Γ_D,"Gamma_D")
  update_labels!(2,model,f_Γ_N,"Gamma_N")
  # Triangulation and measures
  Ω = Triangulation(model)
  Γ_N = BoundaryTriangulation(model,tags="Gamma_N")
  dΩ = Measure(Ω,2*order)
  dΓ_N = Measure(Γ_N,2*order)
  # Spaces
  reffe = ReferenceFE(lagrangian,Float64,order)
  V = TestFESpace(model,reffe;dirichlet_tags=["Gamma_D"])
  U = TrialFESpace(V,0.0)
  V_φ = TestFESpace(model,reffe)
  V_reg = TestFESpace(model,reffe;dirichlet_tags=["Gamma_N"])
  U_reg = TrialFESpace(V_reg,0)
  # Level set and interpolator
  φh = interpolate(lsf_func,V_φ)
  interp = SmoothErsatzMaterialInterpolation(η = 2*maximum(get_el_Δ(model)))
  I,H,DH,ρ = interp.I,interp.H,interp.DH,interp.ρ
  # Weak formulation
  a(u,v,φ,dΩ,dΓ_N) = ∫((I ∘ φ)*κ*∇(u)⋅∇(v))dΩ
  l(v,φ,dΩ,dΓ_N) = ∫(g*v)dΓ_N
  state_map = AffineFEStateMap(a,l,U,V,V_φ,U_reg,φh,dΩ,dΓ_N)
  # Objective and constraints
  J(u,φ,dΩ,dΓ_N) = ∫((I ∘ φ)*κ*∇(u)⋅∇(u))dΩ
  dJ(q,u,φ,dΩ,dΓ_N) = ∫(κ*∇(u)⋅∇(u)*q*(DH ∘ φ)*(norm ∘ ∇(φ)))dΩ
  vol_D = sum(∫(1)dΩ)
  C1(u,φ,dΩ,dΓ_N) = ∫(((ρ ∘ φ) - vf)/vol_D)dΩ
  dC1(q,u,φ,dΩ,dΓ_N) = ∫(-1/vol_D*q*(DH ∘ φ)*(norm ∘ ∇(φ)))dΩ
  pcfs = PDEConstrainedFunctionals(J,[C1],state_map,
    analytic_dJ=dJ,analytic_dC=[dC1])
  # Velocity extension
  α = 4max_steps*γ*maximum(get_el_Δ(model))
  a_hilb(p,q) =∫(α^2*∇(p)⋅∇(q) + p*q)dΩ
  vel_ext = VelocityExtension(a_hilb,U_reg,V_reg)
  # Finite difference scheme
  scheme = FirstOrderStencil(length(el_size),Float64)
  ls_evo = HamiltonJacobiEvolution(scheme,model,V_φ,tol,max_steps)
  # Optimiser
  optimiser = AugmentedLagrangian(pcfs,ls_evo,vel_ext,φh;γ,γ_reinit,
    verbose=true)
  # Solve
  for (it,uh,φh) in optimiser
    data = ["φ"=>φh,"H(φ)"=>(H ∘ φh),"|∇(φ)|"=>(norm ∘ ∇(φh)),"uh"=>uh]
    iszero(it % iter_mod) && writevtk(Ω,path*"out$it",cellfields=data)
    write_history(path*"/history.txt",get_history(optimiser))
  end
  # Final structure
  it = get_history(optimiser).niter; uh = get_state(pcfs)
  writevtk(Ω,path*"out$it",cellfields=["φ"=>φh,
    "H(φ)"=>(H ∘ φh),"|∇(φ)|"=>(norm ∘ ∇(φh)),"uh"=>uh])
end

main("results/")
\end{minted}
\twocolumn

\end{appendices}


\bibliography{main.bib}

\end{document}

%% file: Figure2.tex
\begingroup%
  \makeatletter%
  \providecommand\color[2][]{%
    \errmessage{(Inkscape) Color is used for the text in Inkscape, but the package 'color.sty' is not loaded}%
    \renewcommand\color[2][]{}%
  }%
  \providecommand\transparent[1]{%
    \errmessage{(Inkscape) Transparency is used (non-zero) for the text in Inkscape, but the package 'transparent.sty' is not loaded}%
    \renewcommand\transparent[1]{}%
  }%
  \providecommand\rotatebox[2]{#2}%
  \newcommand*\fsize{\dimexpr\f@size pt\relax}%
  \newcommand*\lineheight[1]{\fontsize{\fsize}{#1\fsize}\selectfont}%
  \ifx\svgwidth\undefined%
    \setlength{\unitlength}{252.57236897bp}%
    \ifx\svgscale\undefined%
      \relax%
    \else%
      \setlength{\unitlength}{\unitlength * \real{\svgscale}}%
    \fi%
  \else%
    \setlength{\unitlength}{\svgwidth}%
  \fi%
  \global\let\svgwidth\undefined%
  \global\let\svgscale\undefined%
  \makeatother%
  \begin{picture}(1,0.9954271)%
    \lineheight{1}%
    \setlength\tabcolsep{0pt}%
    \put(0.27123401,0.70152071){\color[rgb]{0,0,0}\makebox(0,0)[lt]{\lineheight{12}\smash{\begin{tabular}[t]{l}$\Omega$\end{tabular}}}}%
    \put(0,0){\includegraphics[width=\unitlength,page=1]{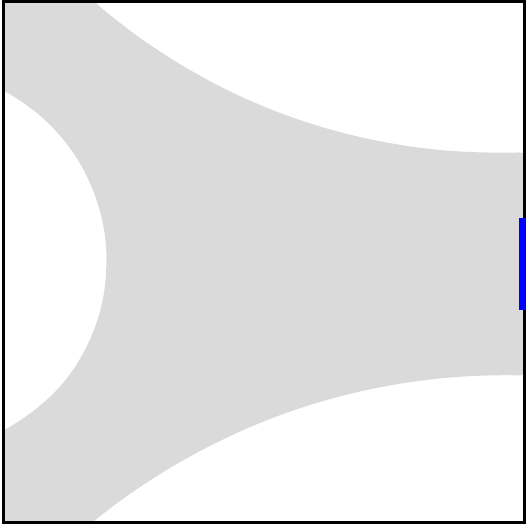}}%
    \put(0.48565109,0.88846041){\color[rgb]{0,0,0}\makebox(0,0)[lt]{\lineheight{12}\smash{\begin{tabular}[t]{l}$D\setminus\Omega$\end{tabular}}}}%
    \put(0.03014374,0.89470504){\color[rgb]{0,0,0}\makebox(0,0)[lt]{\smash{\begin{tabular}[t]{l}$\Gamma_{D}$\end{tabular}}}}%
    \put(0.83010556,0.48834122){\color[rgb]{0,0,0}\makebox(0,0)[lt]{\smash{\begin{tabular}[t]{l}$\Gamma_{N}$\end{tabular}}}}%
    \put(0.04796044,0.0850855){\color[rgb]{0,0,0}\makebox(0,0)[lt]{\smash{\begin{tabular}[t]{l}$\Gamma_{D}$\end{tabular}}}}%
    \put(0,0){\includegraphics[width=\unitlength,page=2]{_Figure2.pdf}}%
  \end{picture}%
\endgroup%

%% file: Figure4a.tex
\begingroup%
  \makeatletter%
  \providecommand\color[2][]{%
    \errmessage{(Inkscape) Color is used for the text in Inkscape, but the package 'color.sty' is not loaded}%
    \renewcommand\color[2][]{}%
  }%
  \providecommand\transparent[1]{%
    \errmessage{(Inkscape) Transparency is used (non-zero) for the text in Inkscape, but the package 'transparent.sty' is not loaded}%
    \renewcommand\transparent[1]{}%
  }%
  \providecommand\rotatebox[2]{#2}%
  \newcommand*\fsize{\dimexpr\f@size pt\relax}%
  \newcommand*\lineheight[1]{\fontsize{\fsize}{#1\fsize}\selectfont}%
  \ifx\svgwidth\undefined%
    \setlength{\unitlength}{296.18659178bp}%
    \ifx\svgscale\undefined%
      \relax%
    \else%
      \setlength{\unitlength}{\unitlength * \real{\svgscale}}%
    \fi%
  \else%
    \setlength{\unitlength}{\svgwidth}%
  \fi%
  \global\let\svgwidth\undefined%
  \global\let\svgscale\undefined%
  \makeatother%
  \begin{picture}(1,0.82285591)%
    \lineheight{1}%
    \setlength\tabcolsep{0pt}%
    \put(0.87220567,0.52907139){\color[rgb]{0,0,0}\makebox(0,0)[lt]{\smash{\begin{tabular}[t]{l}$\Gamma_{N}$\end{tabular}}}}%
    \put(0.09405318,0.16136255){\color[rgb]{0,0,0}\makebox(0,0)[lt]{\smash{\begin{tabular}[t]{l}$\Gamma_{D}$\end{tabular}}}}%
    \put(0.35,0.29467442){\color[rgb]{0,0,0}\makebox(0,0)[lt]{\smash{\begin{tabular}[t]{l}$\Gamma_{D}$\end{tabular}}}}%
    \put(0,0){\includegraphics[width=\unitlength,page=1]{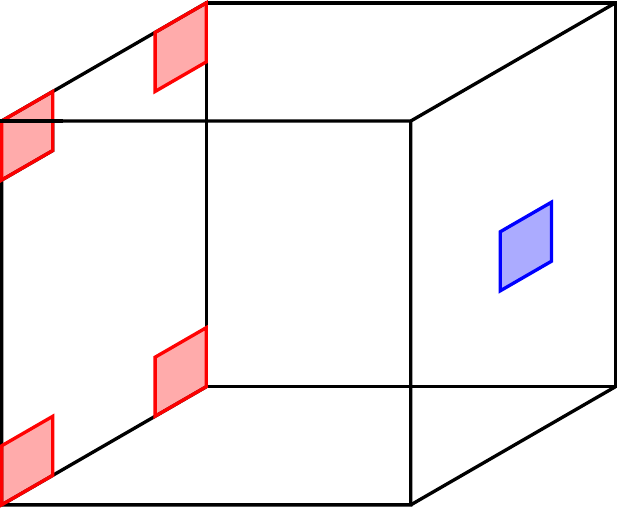}}%
    \put(0.09631437,0.55869981){\color[rgb]{0,0,0}\makebox(0,0)[lt]{\smash{\begin{tabular}[t]{l}$\Gamma_{D}$\end{tabular}}}}%
    \put(0.36401395,0.68801769){\color[rgb]{0,0,0}\makebox(0,0)[lt]{\smash{\begin{tabular}[t]{l}$\Gamma_{D}$\end{tabular}}}}%
  \end{picture}%
\endgroup%

%% file: Figure5.tex
\begingroup%
  \makeatletter%
  \providecommand\color[2][]{%
    \errmessage{(Inkscape) Color is used for the text in Inkscape, but the package 'color.sty' is not loaded}%
    \renewcommand\color[2][]{}%
  }%
  \providecommand\transparent[1]{%
    \errmessage{(Inkscape) Transparency is used (non-zero) for the text in Inkscape, but the package 'transparent.sty' is not loaded}%
    \renewcommand\transparent[1]{}%
  }%
  \providecommand\rotatebox[2]{#2}%
  \newcommand*\fsize{\dimexpr\f@size pt\relax}%
  \newcommand*\lineheight[1]{\fontsize{\fsize}{#1\fsize}\selectfont}%
  \ifx\svgwidth\undefined%
    \setlength{\unitlength}{369.83284674bp}%
    \ifx\svgscale\undefined%
      \relax%
    \else%
      \setlength{\unitlength}{\unitlength * \real{\svgscale}}%
    \fi%
  \else%
    \setlength{\unitlength}{\svgwidth}%
  \fi%
  \global\let\svgwidth\undefined%
  \global\let\svgscale\undefined%
  \makeatother%
  \begin{picture}(1,0.65899742)%
    \lineheight{1}%
    \setlength\tabcolsep{0pt}%
    \put(0.88103493,0.4120708){\color[rgb]{0,0,0}\makebox(0,0)[lt]{\smash{\begin{tabular}[t]{l}$\Gamma_{N}$\end{tabular}}}}%
    \put(0.27865491,0.53964331){\color[rgb]{0,0,0}\makebox(0,0)[lt]{\smash{\begin{tabular}[t]{l}$\Gamma_{D}$\end{tabular}}}}%
    \put(0,0){\includegraphics[width=\unitlength,page=1]{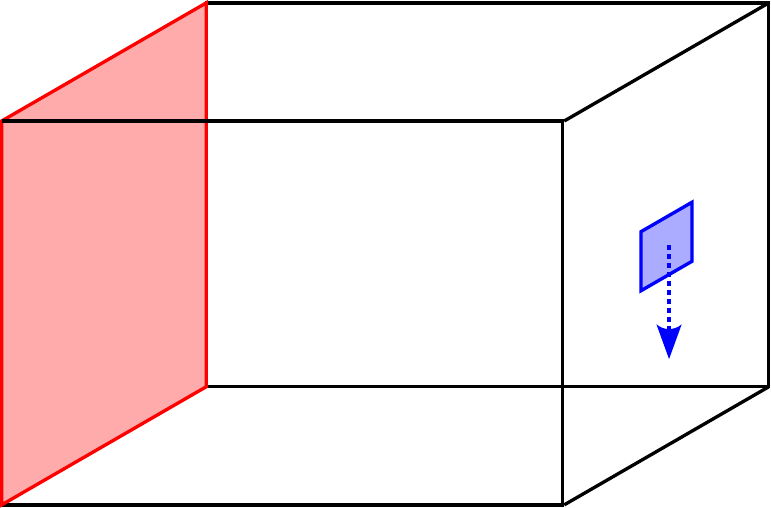}}%
  \end{picture}%
\endgroup%

%% file: Figure7.tex
\begingroup%
  \makeatletter%
  \providecommand\color[2][]{%
    \errmessage{(Inkscape) Color is used for the text in Inkscape, but the package 'color.sty' is not loaded}%
    \renewcommand\color[2][]{}%
  }%
  \providecommand\transparent[1]{%
    \errmessage{(Inkscape) Transparency is used (non-zero) for the text in Inkscape, but the package 'transparent.sty' is not loaded}%
    \renewcommand\transparent[1]{}%
  }%
  \providecommand\rotatebox[2]{#2}%
  \newcommand*\fsize{\dimexpr\f@size pt\relax}%
  \newcommand*\lineheight[1]{\fontsize{\fsize}{#1\fsize}\selectfont}%
  \ifx\svgwidth\undefined%
    \setlength{\unitlength}{315.65957658bp}%
    \ifx\svgscale\undefined%
      \relax%
    \else%
      \setlength{\unitlength}{\unitlength * \real{\svgscale}}%
    \fi%
  \else%
    \setlength{\unitlength}{\svgwidth}%
  \fi%
  \global\let\svgwidth\undefined%
  \global\let\svgscale\undefined%
  \makeatother%
  \begin{picture}(1,0.7962387)%
    \lineheight{1}%
    \setlength\tabcolsep{0pt}%
    \put(0.21136689,0.50259726){\color[rgb]{0,0,0}\makebox(0,0)[lt]{\smash{\begin{tabular}[t]{l}$\Gamma_{\text{in}}$\end{tabular}}}}%
    \put(0.8,0.4981609){\color[rgb]{0,0,0}\makebox(0,0)[lt]{\smash{\begin{tabular}[t]{l}$\Gamma_{\text{out}}$\end{tabular}}}}%
    \put(0.33122942,0.66){\color[rgb]{0,0,0}\makebox(0,0)[lt]{\smash{\begin{tabular}[t]{l}$\Gamma_{D}$\end{tabular}}}}%
    \put(0.08637104,0.17914774){\color[rgb]{0,0,0}\makebox(0,0)[lt]{\smash{\begin{tabular}[t]{l}$\Gamma_{D}$\end{tabular}}}}%
    \put(0.32773104,0.30898758){\color[rgb]{0,0,0}\makebox(0,0)[lt]{\smash{\begin{tabular}[t]{l}$\Gamma_{D}$\end{tabular}}}}%
    \put(0.10262779,0.54){\color[rgb]{0,0,0}\makebox(0,0)[lt]{\smash{\begin{tabular}[t]{l}$\Gamma_{D}$\end{tabular}}}}%
    \put(0,0){\includegraphics[width=\unitlength,page=1]{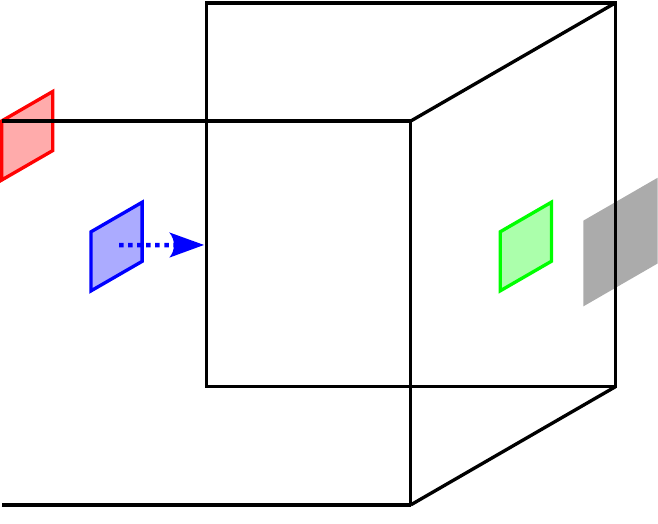}}%
    \put(0,0){\includegraphics[width=\unitlength,page=2]{_Figure7.pdf}}%
  \end{picture}%
\endgroup%

%% file: Figure9a.tex
\begingroup%
  \makeatletter%
  \providecommand\color[2][]{%
    \errmessage{(Inkscape) Color is used for the text in Inkscape, but the package 'color.sty' is not loaded}%
    \renewcommand\color[2][]{}%
  }%
  \providecommand\transparent[1]{%
    \errmessage{(Inkscape) Transparency is used (non-zero) for the text in Inkscape, but the package 'transparent.sty' is not loaded}%
    \renewcommand\transparent[1]{}%
  }%
  \providecommand\rotatebox[2]{#2}%
  \newcommand*\fsize{\dimexpr\f@size pt\relax}%
  \newcommand*\lineheight[1]{\fontsize{\fsize}{#1\fsize}\selectfont}%
  \ifx\svgwidth\undefined%
    \setlength{\unitlength}{296.18659151bp}%
    \ifx\svgscale\undefined%
      \relax%
    \else%
      \setlength{\unitlength}{\unitlength * \real{\svgscale}}%
    \fi%
  \else%
    \setlength{\unitlength}{\svgwidth}%
  \fi%
  \global\let\svgwidth\undefined%
  \global\let\svgscale\undefined%
  \makeatother%
  \begin{picture}(1,0.81889749)%
    \lineheight{1}%
    \setlength\tabcolsep{0pt}%
    \put(0,0){\includegraphics[width=\unitlength,page=1]{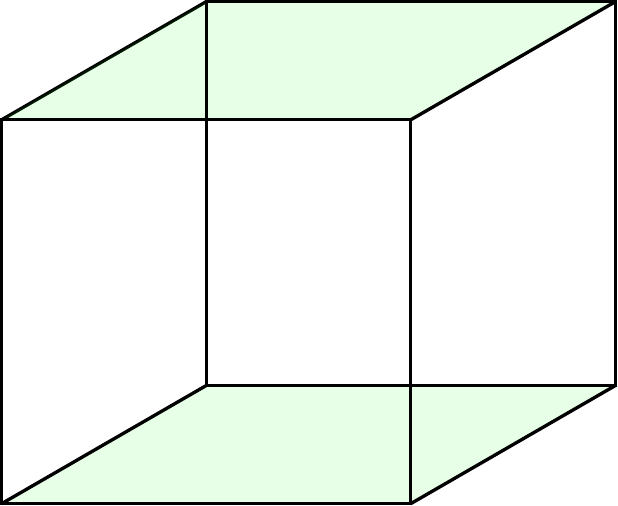}}%
    \put(0.47117386,0.70844416){\color[rgb]{0,0,0}\makebox(0,0)[lt]{\smash{\begin{tabular}[t]{l}$\Gamma_{Z^+}$\end{tabular}}}}%
    \put(0.11320467,0.36362698){\color[rgb]{0,0,0}\makebox(0,0)[lt]{\smash{\begin{tabular}[t]{l}$\Gamma_{X^-}$\end{tabular}}}}%
    \put(0.78441028,0.36362698){\color[rgb]{0,0,0}\makebox(0,0)[lt]{\smash{\begin{tabular}[t]{l}$\Gamma_{X^+}$\end{tabular}}}}%
    \put(0.67114557,0.47452481){\color[rgb]{0,0,0}\makebox(0,0)[lt]{\smash{\begin{tabular}[t]{l}$\Gamma_{Y^+}$\end{tabular}}}}%
    \put(0.34984421,0.29826407){\color[rgb]{0,0,0}\makebox(0,0)[lt]{\smash{\begin{tabular}[t]{l}$\Gamma_{Y^-}$\end{tabular}}}}%
    \put(0.47117386,0.07539667){\color[rgb]{0,0,0}\makebox(0,0)[lt]{\smash{\begin{tabular}[t]{l}$\Gamma_{Z^-}$\end{tabular}}}}%
    \put(0,0){\includegraphics[width=\unitlength,page=2]{_Figure9a.pdf}}%
  \end{picture}%
\endgroup%